\documentclass[a4paper, 11pt]{amsart}
\usepackage{amssymb,array}
\usepackage{url}
\usepackage{graphicx}  
\usepackage{subfigure}
\usepackage{enumerate}

\usepackage[left=1in, right=1in, top=1.5in, bottom=1.5in]{geometry}

\newtheorem{theorem}{Theorem}[section]
\newtheorem{definition}[theorem]{Definition}
\newtheorem{proposition}[theorem]{Proposition}
\newtheorem{remark}[theorem]{Remark}
\newtheorem{lemma}[theorem]{Lemma}

\newtheorem{corollary}[theorem]{Corollary}

\renewcommand{\leq}{\leqslant}
\renewcommand{\geq}{\geqslant}

\newcommand{\R}{\mathbb{R}}
\newcommand{\C}{\mathbb{C}}
\newcommand{\E}{\mathbb{E}}
\renewcommand{\P}{\mathbb{P}}

\newcommand{\U}{\mathcal{U}}
\newcommand{\D}{\mathcal{D}}

\renewcommand{\S}{\mathcal{S}}

\DeclareMathOperator{\trace}{Tr}

\DeclareMathOperator{\id}{id}
\DeclareMathOperator{\Wg}{Wg}
\DeclareMathOperator{\Mob}{Mob}

\DeclareMathOperator*{\Ex}{\mathbb{E}}

\renewcommand{\phi}{\varphi}
\renewcommand{\epsilon}{\varepsilon}

\def\be{\begin{eqnarray}}
\def\ee{\end{eqnarray}}
\def\bee{\begin{equation}}
\def\eee{\end{equation}} 

\begin{document}

\date{\today}

\title[Additivity property and conjugate pairs of quantum channels]{Asymptotically well-behaved input states do not violate additivity for conjugate pairs of random quantum channels}

\author{Motohisa Fukuda}
\address{Zentrum Mathematik, M5,
Technische Universit\"{a}t M\"{u}nchen,
Boltzmannstrasse 3,
85748 Garching, Germany}
\email{m.fukuda@tum.de}
\author{Ion Nechita}
\address{CNRS, Laboratoire de Physique Th\'eorique, IRSAMC, Universit\'e de Toulouse, UPS, 31062 Toulouse, France}
\email{nechita@irsamc.ups-tlse.fr}

\begin{abstract}
It is now well-known that, with high probability, the additivity of  minimum output entropy does not hold
for a pair of a random quantum channel and its complex conjugate. 
We investigate asymptotic behavior of output states of $r$-tensor powers of such pairs, as the dimension of inputs grows. 
We compute the limit output states for any sequence of well-behaved inputs,
which consist of a large class of input states having a nice set of parameters.  
Then, we show that among these input states tensor products of Bell states give asymptotically the least output entropy, 
giving positive mathematical evidence towards additivity of above pairs of channels. 
\end{abstract}
 
\maketitle

\section{Introduction}  
In this paper, we investigate the asymptotic limits of output states of $(\Phi \otimes \bar \Phi)^{\otimes r}$,
where $\Phi$ is a random quantum channel whose input space grows. 
This kind of pair of a quantum channel and its complex conjugate, $\Phi \otimes \bar \Phi$, which we shall call \emph{conjugate pair}, has been known to violate the additivity of minimum output entropy in the following sense: 
\[
S_{\min} (\Phi \otimes \bar \Phi) < S_{\min} (\Phi) + S_{\min}(\bar \Phi) 
\]
This is generically true when the dimensions of concerned spaces are large in a certain regime. 
Here, the minimum output entropy $S_{\min}(\cdot)$ is defined as 
$S_{\min} (\Phi) = \min_{\rho} S(\Phi(\rho))$,
where $S(\cdot)$ is the von Neumann entropy and $\rho$ ranges over all the inputs. 
This violation of additivity was shown first by Hastings \cite{hastings} for random unitary channels (see also \cite{fukuda-king-moser}) and then later the violation was proven for general quantum channels
\cite{brandao-horodecki,fukuda-king,aubrun-szarek-werner}. 
The additivity question of minimum output entropy was raised by King and Ruskai in \cite{king-ruskai} 
for any pair of quantum channels, and 
it attracted more attention when Shor proved its equivalence to additivity question of Holevo capacity in \cite{shor}.   

One of the interesting facts which yield this violation is that 
any conjugate pair $\Phi \otimes \bar \Phi$ has an output with a rather large eigenvalue
for a Bell-state input. 
This phenomenon was pointed out first in \cite{hayden-winter} to disprove 
the additivity of minimal output $p$-Renyi entropy for $p>1$ (see also \cite{cn2,bcn}), and
the precise limit eigenvalue distribution was calculated in \cite{cn1,cn3},  using random matrix theory. 
A search for the optimal input for this conjugate pair was carried out in \cite{cfn1,cfn2} and 
it was shown that  among some large classes of input states, a Bell state asymptotically gives the least  output entropy through the random conjugate pair $\Phi \otimes \bar \Phi$.
This of course does not imply that 
a Bell state actually gives the least output entropy but
it gives solid mathematical evidence for the physical intuition that
a Bell state is the optimal input for small output entropy.

In this paper, we ask the same kind of question: 
``what is the optimal input for  $\Psi^{\otimes r}$?'', where $\Psi =\Phi \otimes \bar \Phi$.
This question is related to the additivity question for $r$-tensor power of the conjugate pair:
\[
S_{\min} (\Psi^{\otimes r}) \stackrel{?}{=} r S_{\min} (\Psi)
\]
which was positively conjectured in \cite{hastings};
the paper argues an intuition that for low output entropy 
only bipartite entanglement between $\Phi$ and $\bar\Phi$ is useful while
multi-partite entanglement is not. 
We enforce this intuition mathematically by showing that
among large classes of inputs products of Bell states asymptotically give the least output entropy
where bipartite spaces for these Bell states make pairs of $\Phi$ and $\bar\Phi$. 
 
The novelty of our results consists in the fact that we are considering arbitrary tensor powers of $\Psi$: $\Psi^{\otimes r}$, whereas previous work \cite{hayden-winter, hastings, cn1, cn3, cfn1,cfn2}  dealt with the case $r=1$.
Although some weak form of additivity for $\Phi^{\otimes r}$, where $\Phi$ is the random quantum channel,
was shown in \cite{montanaro}, our results are different in the following sense. 
We specify the limiting output matrix of $\Psi^{\otimes r}$ for any well-behaved inputs
which have a set of stable parameters as the dimension of system grows,
in order to discuss on optimal inputs. 
The main result of the paper can be stated informally as follows: 
\begin{theorem}
Let $\Phi_n$ be sequence of random quantum channels defined by random Haar isometries and put $\Psi_n = \Phi_n \otimes \bar \Phi_n$. For any $r \geq 1$, any sequence of ``well-behaved'' input states for $\Psi_n^{\otimes r}$ yields an asymptotical output entropy larger than products of $r$ Bell states. In other words, ``well-behaved'' inputs  can not violate the additivity relation for $\Psi_n$. 
\end{theorem}

To prove the above statement, we introduce new techniques to study high tensor powers of random matrices. Not surprisingly, we have to deal with operators corresponding to the action of the symmetric group on tensor powers of vector spaces, due to the use of the Schur-Weyl duality via the Weingarten formula. 
First, we perform a detailed spectral analysis of these operators, needed to state optimality result for the von Neumann entropy. 
Secondly, in the appendix, we compute bounds for generalized traces of arbitrary matrices acting on tensor powers, in terms of the $L^1$, $L^2$ or the $L^\infty$ norms of the matrices. We believe these two technical results to be important on their own and to find applications in future work on the subject.

The paper is organized as follows. 
In Section \ref{sec:graphical-calculus}
we briefly explain the setting and the tools: random quantum channels and graphical Weingarten calculus. Then, in Section \ref{sec:partial-permutations}, 
we analyze some operators defined via permutations, 
which are used later. 
After calculating the limiting output matrix of $\Psi^{\otimes r}$ for well-behaved inputs in Section \ref{sec:main}, we conclude in Section \ref{sec:Bell-optimal} that among these,
tensor products of Bell states yield the least output entropy.
In Section \ref{sec:GHZ} we show that GHZ states and generic multi-partite inputs asymptotically give the maximally mixed outputs. 
Finally, in Appendix \ref{sec:trace-bounds}, we obtain
bounds for generalized traces of matrices, which are used to prove results in Section \ref{sec:main}.

\section{Random quantum channels and the graphical Weingarten formula}\label{sec:graphical-calculus}

This section is split into three parts, which treat, in order, random quantum channels, the algebraic Weingarten formula, and its graphical implementation.

\subsection*{Random quantum channels}
A quantum channel $\Phi:  M_n(\mathbb C) \to M_n(\mathbb C)$ is a linear, completely positive and trace preserving map \cite{nielsen-chuang}. By the Stinespring dilation theorem, such a map can be written as
\be
\Phi (X) = \trace_{\C^k}(V X V^*)
\ee 
where
\be
V : \C^n \rightarrow \C^k\otimes \C^n 
\ee  
is an isometry, $V^*V = I$. 
One can define a probability measure on the set of quantum channels starting from the Haar measure on the unitary group $\U(kn)$ in the following way. We endow the set of isometries $\C^n \rightarrow \C^k\otimes \C^n $ by truncating a unitary matrix $U$ distributed along the Haar measure and we consider the image measure on the set of quantum channels. In other words, 
\be\label{picture by U}
\Phi (X) = \trace_{\C^k} ( U (P_e \otimes X ) U^* ),
\ee 
where $P_e \in M_k(\C)$ is a rank one projector, called the state of the environment.

If one switches the roles of output and environment spaces, we obtain the \emph{complementary channel}
$\Phi^C: M_n(\mathbb C) \to M_k(\mathbb C)$
\cite{Holevo05,KMNR}:
\be
\Phi^C (X) = \trace_{\C^n} (V X V^*)
\qquad \text{or} \qquad 
\Phi^C (X) = \trace_{\C^n} (U (P_e \otimes \rho ) U^*)
\ee 
A quantum channel and its complementary channel share 
the same non-zero output eigenvalues for any rank one input state, see  \cite{Holevo05,KMNR} for details. We shall prefer working with complementary channels because we are interested in the asymptotic regime $n \to \infty$ and $k$ fixed.
 
\subsection*{Unitary integration} 
  
Our random matrices are built from Haar random unitary operators, so we are interested in computing averages with respect to this measure. The computation of averages of monomials in the entries of a random unitary matrices has been performed in \cite{collins-imrn, collins-sniady}.

\begin{theorem}
\label{thm:Wg}
 Let $n$ be a positive integer and
$i=(i_1,\ldots ,i_p)$, $i'=(i'_1,\ldots ,i'_p)$,
$j=(j_1,\ldots ,j_p)$, $j'=(j'_1,\ldots ,j'_p)$
be $p$-tuples of positive integers from $\{1, 2, \ldots, n\}$. Then
\begin{multline}
\label{bid} \int_{\U(n)}U_{i_1j_1} \cdots U_{i_pj_p}
\overline{U_{i'_1j'_1}} \cdots
\overline{U_{i'_pj'_p}}\ dU=\\
\sum_{\alpha, \beta\in \S_{p}}\delta_{i_1i'_{\alpha (1)}}\ldots
\delta_{i_p i'_{\alpha (p)}}\delta_{j_1j'_{\beta (1)}}\ldots
\delta_{j_p j'_{\beta (p)}} \Wg (n,\beta^{-1}\alpha),
\end{multline}
where the function $\Wg$ is called the Weingarten function \cite{collins-sniady}.
If $p\neq p'$ then
\begin{equation} \label{eq:Wg_diff} \int_{\U(n)}U_{i_{1}j_{1}} \cdots
U_{i_{p}j_{p}} \overline{U_{i'_{1}j'_{1}}} \cdots
\overline{U_{i'_{p'}j'_{p'}}}\ dU= 0.
\end{equation}
\end{theorem}
  
  Let us recall the definition of the unitary Weingarten function.
\begin{definition}
The unitary Weingarten function 
$\Wg(n,\sigma)$
is a function of a dimension parameter $n$ and of a permutation $\sigma$
in the symmetric group $\S_p$. 
It is the inverse of the function $\sigma \mapsto n^{\#  \sigma}$ under the convolution for the symmetric group algebra ($\# \sigma$ denotes the number of cycles of the permutation $\sigma$):
\begin{equation}
\forall \sigma,\pi \in \S_p, \quad \sum_{\tau \in \S_p} \Wg(n,\sigma^{-1}\tau)n^{\#(\tau^{-1}\pi)} = \delta_{\sigma, \pi}.
\end{equation}
It has the following asymptotics 
\begin{equation}\label{eq:Wg-asympt}
\Wg(n,\sigma) = n^{-(p + |\sigma|)} (\Mob(\sigma) + O(n^{-2})),
\end{equation}
where the M\"{o}bius function is multiplicative on the cycles of $\sigma$ and its value on a $r$-cycle is
$$(-1)^{r-1} \mathrm{Cat}_{r-1},$$
where $\mathrm{Cat}$ are the Catalan numbers and $|\sigma |$ is the \emph{length} of $\sigma$, i.e. the minimal number of transpositions that multiply to $\sigma$.
\end{definition}

Finally, note that the function $|\sigma| = p-\#\sigma$ induces a distance on $\S_p$, $d(\sigma, \pi) = |\sigma^{-1} \pi|$.

\subsection*{Graphical Weingarten calculus}

We recall now the graphical Weingarten method for computing unitary integrals, introduced in \cite{cn1}. We shall only give the main ideas, referring the reader to \cite{cn1} for all the details. Recent work making heavy use of this technique is \cite{cn3,cn-entropy,cnz,cfn1,aubrun-nechita,cfn2}.

The Weingarten graphical calculus builds up on the \emph{tensor diagrams} introduced by physicists and adds to it the ability to perform averages over unitary elements. In the graphical formalism, tensors (vectors, linear forms, matrices, etc) are represented by \emph{boxes}. To each box, one attaches labels of different shapes, corresponding to vector spaces. The labels can be empty (white) or filled (black), corresponding to spaces or their duals: a $(p,q)$-tensor will be represented by a box with $p$ white labels and $q$ black labels attached. 

Besides boxes, our diagrams contain \emph{wires}, which connect the labels attached to boxes. Each wire corresponds to a tensor contraction between a vector space $V$ and its dual $V^*$: $V \times V^* \to \mathbb C$. A \emph{diagram} is a collection of such boxes and wires and corresponds to an element in a tensor product space, see Figure \ref{fig:A}.

\begin{figure}[htbp] 
\includegraphics{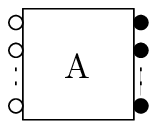} \qquad\qquad\qquad \includegraphics{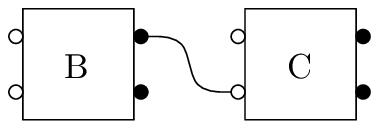}
\caption{Diagram for a matrix acting on a tensor product space and for a tensor contraction between matrices $B$ and $C$.} 
\label{fig:A}
\end{figure}

Let us now describe briefly how one computes expectation values of such diagrams containing random Haar unitary boxes. The idea in \cite{cn1} was to implement in the graphical formalism the Weingarten formula in Theorem \ref{thm:Wg}. Each pair of permutations $(\alpha,\beta)$ in  (\ref{bid}) will be used to eliminate $U$ and $\overline U$ boxes and wires will be added between the white, resp. black, labels of the box $U$ with index $i$ and the white, resp. black, labels of the box $\bar U$ with index $\alpha(i)$, resp. $\beta(i)$. In this way, for each pair of permutations, one obtains a new diagram, called a \emph{removal} of the original diagram. The graphical Weingarten formula is described in the following theorem \cite{cn1}.

\begin{theorem}\label{thm:graphical-Wg}
If $\mathcal D$ is a diagram containing boxes $U, \overline U$ corresponding to a Haar-distributed random unitary matrix $U \in \mathcal U(n)$, the expectation value of $\mathcal D$ with respect to $U$ can be decomposed as a sum of removal diagrams $\D_{\alpha, \beta}$, weighted by Weingarten functions:
\[\E_U(\D)=\sum_{\alpha, \beta} \D_{\alpha, \beta} \Wg (n, \alpha^{-1}\beta).\]
\end{theorem}

\section{Partial permutations and their associated operators}\label{sec:partial-permutations}

In this section we define \emph{partial permutation matrices}, central objects which will appear several times in the paper. 
Then, we construct several matrices out of partial permutation matrices,
which are used later in this paper. 
Readers may just look over this section to get familiar with newly defined matrices and
come back to it whenever needed.  

For a permutation $\sigma \in \mathcal S_r$, 
we define $P^\otimes_\sigma \in M_{n^r}(\mathbb C)$ the \emph{tensor permutation matrix} on simple tensors by
\be\label{def:permutation-matrix} 
\forall x_i \in \mathbb C^n, \qquad P^\otimes_\sigma(x_1 \otimes \cdots \otimes x_r) = x_{\sigma^{-1}(1)} \otimes \cdots \otimes  x_{\sigma^{-1}(r)}.
\ee
These operators correspond to the usual action $\S_r \curvearrowright (\C^n)^{\otimes r}$. Note that we use the tensor superscript to distinguish these operators from the usual permutation matrices $P_\sigma$ which act as $P_\sigma e_i = e_{\sigma^{-1}(i)}$ on a basis $\{e_1, \ldots, e_n\}$ of $\C^n$.

A \emph{partial permutation} on $\{1,2,\ldots,r\}$ is defined to be an injective map
$$\alpha: \mathrm{dom}(\alpha) \to \{1,2,\ldots,r\},$$
where $\mathrm{dom}(\alpha)$ stands for the \emph{domain} of $\alpha$. The set of all partial permutations of $ \{1,2,\ldots,r\}$ is denoted by $\hat \S_r$. We denote by $\emptyset$ the partial permutation with empty domain. The cardinality of $\hat \S_r$ is
\be\label{formula:card-pp}
|\hat \S_r| = \sum_{k=0}^r \left[\binom{r}{k}\right]^2k !
\ee
We endow the set $\hat \S_r$ with the ``map extension'' partial order: $\alpha \leq \beta$ iff 
$$\mathrm{dom}(\alpha) \subseteq \mathrm{dom}(\beta) \quad \text{ and } \quad \alpha(i) = \beta(i) \, \forall i \in \mathrm{dom}(\alpha).$$
For this partial order, the empty partial permutation $\emptyset$ is the unique minimal element and every ``full'' permutation $\sigma \in \S_r \subset \hat \S_r$ is a maximal element. 

Suppose we have a $2r$-partite space $(\C^n)^{\otimes r} \otimes (\C^n)^{\otimes r}$ and 
name the spaces as follows:
$$(\C^n)^{\otimes r} \otimes (\C^n)^{\otimes r} = \C^n_{[1,T]} \otimes \C^n_{[2,T]} \otimes \cdots \otimes  \C^n_{[r,T]} \otimes \C^n_{[1,B]} \otimes \C^n_{[2,B]}\otimes \cdots  \otimes \C^n_{[r,B]},$$
where $T$ stands for ``top'' and $B$ stands for ``bottom'' (the spaces should be imagined stacked vertically, see Figure \ref{fig:partial-perm} for the case $r=2$).
For a partial permutation $\alpha \in \hat S_r$ define
\be\label{operator:T}
T_\alpha^{(n)} = \bigotimes _{x \in {\rm dom} (\alpha)} B_{x,\alpha (x)}^{(n)} \otimes I 
\ee
where $B^{(n)}_{x,y} \in \C^n_{[x,T]} \otimes \C^n_{[y,B]}$ is an un-normalized Bell state:
\be
B_{x,y}^{(n)} = \sum_{i,j=1}^n e_i e_j^* \otimes f_if_j^*  
\ee
where $\{e_i\}$ and $\{f_j\}$ are the canonical bases for the spaces with labels $[x,T]$ and $[y,B]$, respectively. In other words,
\begin{enumerate}[i)]
\item If $x \in {\rm dom} (\alpha)$, then $T_\alpha^{(n)}$ acts on $\C^n_{[x, T]} \otimes \C^n_{[\alpha(x),B]}$ as the un-normalized Bell state;
\item If $x \not \in {\rm dom} (\alpha)$, then $T_\alpha^{(n)}$ acts on $\C^n_{[x, T]}$ as the identity;
\item If $y \neq \alpha(x)$ for any $x  \in {\rm dom} (\alpha)$, then $T_\alpha^{(n)}$ acts on $\C^n_{[y, B]}$ as the identity.
\end{enumerate} 

In Figure \ref{fig:partial-perm}, we use the graphical notation from Section \ref{sec:graphical-calculus} to represent the seven different operators $T_\alpha^{(n)}$ in the case $r=2$ (note that $7=1+4+2$ as in \eqref{formula:card-pp}). 
 
\begin{figure}[htbp]
\includegraphics{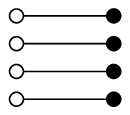}\qquad
\includegraphics{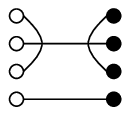}\qquad
\includegraphics{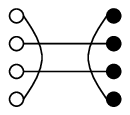}\qquad
\includegraphics{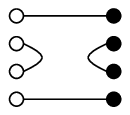}\qquad
\includegraphics{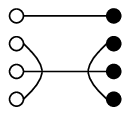}\qquad
\includegraphics{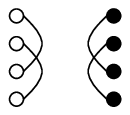}\qquad
\includegraphics{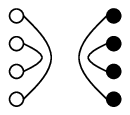}
\caption{Diagrams for the operators $T_\alpha^{(n)}$ in the case $r=2$} 
\label{fig:partial-perm}
\end{figure}
 
We also define the normalized versions of the above operators:
\be\label{operator_rescaled}
\tilde B_{x,y}^{(n)} = \frac{1}{n} B_{x,y}^{(n)}; \qquad 
\tilde T_\alpha^{(n)} = \bigotimes _{x \in {\rm dom} (\alpha)} \tilde B_{x,\alpha (x)}^{(n)} \otimes I,
\ee 
which are projections. 

In the rest of this section, we compute (asymptotically) the spectrum of an operator appearing in the study of the output states of Theorem \ref{thm:main}. The main result here will be used in Section \ref{sec:Bell-optimal} to analyze the optimality of input states for random quantum channels.

For a partial permutation $\alpha \in \hat{\mathcal S}_r$, define the operator
\begin{equation}\label{eq:def-Q-alpha}
Q_\alpha^{(n)} = \sum_{\beta \geq \alpha}\tilde T_{\beta}^{(n)}  (-1)^{|{\rm dom}\beta|-|{\rm dom}\alpha|} 
\end{equation}
Our goal (see Proposition \ref{prop:Z-conv-comb}) is to show that the matrices $Q_\alpha^{(n)} $ are asymptotically positive (i.e. their eigenvalues converge to non-negative numbers) when $n \to \infty$. We shall do more than this, by showing that their asymptotic spectrum is the set $\{0,1\}$. First, note that any partial permutation $\beta$ satisfying $\beta \geq \alpha$ can be written as a ``direct sum'' $\beta = \alpha \oplus \gamma$, where $\gamma$ can be any partial permutations such that the domains and the images of $\alpha$ and $\gamma$ are disjoint: $\mathrm{dom}(\alpha) \cap \mathrm{dom}(\gamma) = \emptyset$ and  $\alpha(\mathrm{dom}(\alpha)) \cap \gamma(\mathrm{dom}(\gamma)) = \emptyset$. Using this observation, one can write
\begin{equation}\label{eq:decomp-Q-alpha}
Q_\alpha^{(n)} = \left[ \bigotimes _{x \in {\rm dom} (\alpha)} \tilde B^{(n)}_{x,\alpha (x)} \right] \otimes \hat Q_\emptyset^{(n)},
\end{equation}
where $\hat Q_\emptyset^{(n)}$ is a copy of $ Q_\emptyset^{(n)}$ with $r$ replaced by $\hat r = r - |\mathrm{dom}(\alpha)|$. From this decomposition, we see that the operator $Q_\emptyset^{(n)}$ plays a special role, as it appears in every other $Q_\alpha^{(n)}$. Moreover, it is clear that $Q_\alpha^{(n)}$ is asymptotically positive iff $Q_\emptyset^{(n)}$ is. 

Since we expect the operator $Q_\emptyset^{(n)} \in M_{n^{2r}}(\mathbb C)$ to have eigenvalues with high multiplicity, we study it as a left multiplication operator on an algebra $\mathcal A_n \ni 1$. Besides the operators $T_\alpha^{(n)}$ (which connect ``top'' spaces with ``bottom'' spaces), the algebra $\mathcal A_n$ contains all tensor permutation operators acting separately on the top and the bottom spaces. Formally, the $\mathcal A_n$ is defined as the algebra generated inside $M_{n^{2r}}(\mathbb C)$ by the operators $T_\alpha^{(n)}$ and the tensor permutation operators of the form $P^\otimes_{\pi_T} \otimes P^\otimes_{\pi_B}$, for any permutations $\pi_{T,B} \in \mathcal S_r$:
$$\mathcal A_n = \mathrm{alg}\left( \{T_\alpha^{(n)}\}_{\alpha \in \hat{\mathcal S}_r} \cup \{P^\otimes_{\pi_T} \otimes P^\otimes_{\pi_B}\}_{\pi_{T,B} \in \mathcal S_r}\right).$$

Up to normalization constants (depending on $n$), the algebra $\mathcal A_n$  has a combinatorial structure, induced by the wire contractions. It can be easily seen that the product $T_\alpha^{(n)}T_\beta^{(n)}$ can be written as $cT_\gamma^{(n)} P^\otimes_{\pi_T} \otimes P^\otimes_{\pi_B}$ for some partial permutation $\gamma \in  \hat{\mathcal S}_r$, permutations $\pi_T,\pi_B \in \mathcal S_r$ and a constant $c>0$ (actually, $c$ is a positive power of $n$). We introduce next the following subsets:
\begin{itemize} 
\item $\mathcal X_1 = \{T_\alpha^{(n)}: \alpha \in \hat S_r\}$
\item $\mathcal X_2 = \{P^\otimes_{\pi_T} \otimes P^\otimes_{\pi_B}: 
\pi_{T},\pi_{B} \in S_r \}\setminus \mathcal X_1$: permutation matrices not of the above type;
\item $\mathcal X_3 = \{ T_\alpha^{(n)}[ P^\otimes_{\pi_T} \otimes P^\otimes_{\pi_B}]: \alpha \in \hat S_r
\text{ and } \pi_{T}, \pi_{B}\in S_r\}\setminus(\mathcal X_1 \sqcup \mathcal X_2)$: products  not of the above types.
\end{itemize}
The set $\mathcal X = \mathcal X_1 \sqcup \mathcal X_2 \sqcup \mathcal X_3$ admits the following elegant description. 
\begin{lemma}\label{lemma:X}
The set $\mathcal X$ is in bijection with the set of permutations of $2r$ objects. More precisely, the map
\begin{align*}
 \mathcal{S}_{2r} & \to \mathcal X\\
\sigma & \mapsto (\mathrm{id} \otimes \mathrm{t})P^\otimes_\sigma
\end{align*}
is a bijection, where $\mathrm{t}$ denotes the transposition operator on $M_{n^r}(\mathbb C)$ acting on the ``bottom'' subspaces.
\end{lemma} 
 
\begin{figure}[htbp]
\subfigure[]{\includegraphics{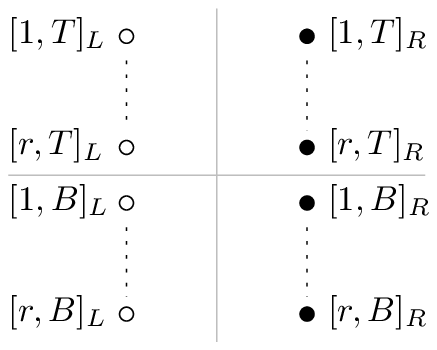}}\qquad
\subfigure[]{\includegraphics{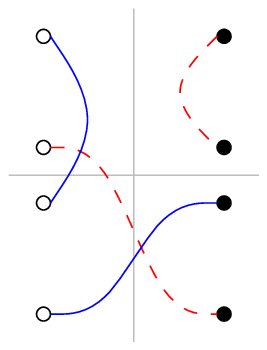}}\qquad
\subfigure[]{\includegraphics{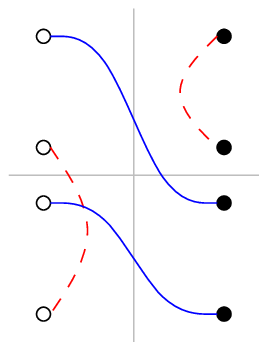}}
\caption{Graphical representation of operators from $\mathcal X$: (a) the labeling of the spaces, (b) possible (blue) and impossible (red, dashed) wires for an element of $\mathcal X$, (c) possible and impossible wires for an element of $(\mathrm{id} \otimes \mathrm{t}) \mathcal X$.} 
\label{fig:mathcal-X}
\end{figure}

\begin{proof}
Graphically, elements in $\mathcal X$ are collection of wires (tensor contractions) connecting $4r$ points (vector spaces): on the left hand side $[x,T/B]_L$ (corresponding to primal spaces) and on the right hand side $[y,T/B]_R$ (corresponding to dual spaces), see Figure \ref{fig:mathcal-X}. By definition, $\mathcal X$ contains all such diagrams, with the exception of the ones connecting some $[x,T]_L$ to some $[y,B]_R$, or some $[x,B]_L$ to some $[y,T]_R$. The partial transposition operation has the following effect on the diagram of an element of $\mathcal X$: it swaps the bottom left and right points, i.e. $[x,B]_L \leftrightarrow [x,B]_R$. Hence, $(\mathrm{id} \otimes \mathrm{t}) \mathcal X$ contains all wire diagrams, such that all wires are ``crossing'', i.e. points on the left are connected to points on the right. But these are exactly the diagrams for tensor permutations of $\S_{2r}$ which concludes the proof.
\end{proof}

\begin{lemma}\label{lemma:basis}
The set $\mathcal X$ is a basis for the algebra $\mathcal A_n$, for all $n \geq 2r$. 
\end{lemma} 
\begin{proof}
It is clear from the definitions that the elements in $\mathcal X$ span the algebra $\mathcal A_n = \mathrm{span}_\C \mathcal X$. Using Lemma \ref{lemma:X}, we need to show that the family of tensor permutations matrices $\{P^\otimes_\sigma\}_{\sigma \in \S_{2r}}$ is linearly independent in $M_{n^{2r}}(\C)$. To this end, consider a linear combination
$$\sum_{\sigma \in S_{2r}} a_\sigma P^\otimes_\sigma = 0$$
Fix some permutation $\pi \in \S_{2r}$. Since $n \geq 2r$, we can pick an orthonormal family $\{e_1, \ldots, e_{2r}\}$ in $\C^n$. Let $x =  e_{\pi(1)} \otimes e_{\pi(2)} \otimes \cdots \otimes e_{\pi(2r)}$ and $y = e_1 \otimes e_2 \otimes \cdots \otimes e_{2r}$. It is easy to check that $\langle y, P^\otimes_\sigma x \rangle = \delta_{\pi,\sigma}$, hence $a_\pi = 0$. This shows that the linear combination must be trivial and thus the tensor permutation matrices are linearly independent.
\end{proof} 

\begin{remark}
Notice that usual permutation matrices are not linear independent, as one can see from the following relation holding in $\S_4$: $P_{(12)}+P_{(34)} = P_{(12)(34)} + P_\mathrm{id}$.
\end{remark}

Write $Q=Q_\emptyset^{(n)} $ (see \eqref{eq:def-Q-alpha}) and
define the left multiplication operator on $\mathcal A_n$ as follows: 
\begin{align*}
L_Q : \mathcal A_n &\rightarrow  \mathcal A_n \\
A &\mapsto  Q A 
\end{align*}
\begin{lemma}\label{lemma:spectral}
One has the following inclusion of spectra
\[
{\rm spec} (Q) \subseteq {\rm spec} (L_Q)
\]
\end{lemma} 
\begin{proof}
Since $Q$ is Hermitian, write
\[
Q = \sum_{i=1}^k \lambda_i P_i
\]
to be the spectral decomposition of $Q$, where $\lambda_i \in \R$ and 
$P_i$ are projections on the eigenspaces. 
Then, 
\[
(\lambda_i \id - L_Q )P_i =  (\lambda_i I - Q )P_i =0
\]
On the other hand,
\[
P_i \in {\rm alg} \{Q\} \subseteq \mathcal A_n 
\]
Hence, $\lambda_i \id - L_Q$ is not invertible on $\mathcal A_n$ and thus $\lambda_i \in  {\rm spec} (L_Q)$.  
\end{proof} 

The algebra $\mathcal A_n$ has a combinatorial nature, and $\mathrm{dim} (\mathcal A_n) = (2r)!$. Notice that the dimension of $\mathcal A_n$ does not grow with $n$, fact which is crucial for the analysis that follows. The previous lemma relates the spectrum of the operator $Q=Q_\emptyset^{(n)}$ (a high dimensional object) to the spectrum of $L_Q$ which is of bounded dimension. We now state the main result of this section. 
\begin{theorem}\label{thm:Q-asympt-positive}
For all $r \geq 1$, the spectrum of the matrix $Q=Q_\emptyset^{(n)}$ is at a distance $O(1/n)$ from the set $\{0,1\}$.
\end{theorem}
\begin{proof}
By the inclusion of spectra proved in Lemma \ref{lemma:spectral}, it suffices to show the conclusion for the operator $L_Q$. By doing this, we are working with a matrix of fixed size, the dependence in $n$ appearing in the entries of $L_Q$. 
For computational simplicity, we rescale the basis $\mathcal X$ into $\tilde{\mathcal X}$, 
in the sense of \eqref{operator_rescaled},
so that all the elements in $\tilde{\mathcal X}$ are projections, up to multiplication by $P_{\pi_T} \otimes P_{\pi_B}$.  
Based on Lemma \ref{lemma:basis},
we write the operator $L_Q$ in the basis $\tilde{\mathcal X}$.
\[
L_Q = 
 \bordermatrix{ 
~ & 1  & 2 &3 \cr
1 & L_{1,1} & L_{1,2}&L_{1,3} \cr
2 & L_{2,1} & L_{2,2} &L_{2,3}\cr
3 & L_{3,1} &L_{3,2}& L_{3,3}\cr}
\]
The indices above indicate which spaces are associated;
$i$ corresponds to the space spanned by $\mathcal X_i$. First, we analyze the blocks $L_{1,1}$, $L_{2,1}$ and $L_{3,1}$. We claim that
\begin{equation}\label{eq:L-1-1}
\tilde T_\alpha^{(n)} \tilde T_\beta^{(n)} =  \begin{cases} 
\tilde T_{\alpha \vee \beta}^{(n)} & \text{if $\alpha \vee \beta$ exists in the poset }(\hat S_r, \leq)\\
 O(n^{-1}) & \text{otherwise} 
\end{cases} 
\end{equation}
Here, we understand $O(n^{-1})$ to be an element of norm $O(n^{-1})$ from $\mathcal A_n$. 
Indeed, one has that
\[
\left (\tilde B_{x,y} \otimes I \right)\cdot \left(\tilde B_{z,w} \otimes I\right)  = 
\begin{cases}
\tilde B_{x,y} \otimes I & \text{if $x=z$ and $y=w$} \\
\tilde B_{x,y} \otimes \tilde B_{z,w} \otimes I & \text{if $x\not = z$ and $y \not = w$ } \\
O(n^{-1}) & \text{otherwise}
\end{cases} 
\]
This follows from the fact that in order to obtain non-vanishing elements when concatenating $\tilde B_{x,y}$ and $ \tilde B_{z,w} $, the half loops have to meet either the identity matrix or the exact same half loop. Taking products of the previous relations yields \eqref{eq:L-1-1}.
Then, 
\begin{align}\label{calculation:spectral}
Q \tilde T_\beta^{(n)} = Q^{(n)}_\emptyset \tilde T_\beta^{(n)} 
&=  \sum_{\alpha} (-1)^{|{\rm dom}(\alpha)|}  \tilde T_{\alpha}^{(n)} \tilde T_{\beta}^{(n)}  \\ 
&=  \sum_{\substack{ \alpha_1 \leq \beta \\ {\rm dom} (\alpha_2) \cap {\rm dom} (\beta) = \emptyset 
\\\alpha_2({\rm dom} (\alpha_2)) \cap \beta({\rm dom} (\beta)) = \emptyset}} 
(-1)^{|{\rm dom} (\alpha_1)|+ |{\rm dom} (\alpha_2)| } 
\tilde T_{\alpha_2}^{(n)}  \tilde T_{\beta}^{(n)}  + O(n^{-1}) \notag\\ 
&= \begin{cases}
Q + O(n^{-1}) & \text{if $\beta = \id$}\notag \\ 
O(n^{-1}) & \text{otherwise}
\end{cases}
\end{align}
In the second equality, we gather all the vanishing terms in the $O(n^{-1})$ such that the surviving partial permutations $\alpha$ are of the form $\alpha_1 \oplus \alpha_2$. For the last equality, we used:
\[
\sum_{\alpha_1 \leq \beta} (-1)^{|{\rm dom} (\alpha_1)|} = (1-1)^{|{\rm dom} (\beta)|} = \delta_{\beta,\mathrm{id}}
\]
This implies that 
\[
L_{1,1} = 
\begin{pmatrix}
\begin{matrix}  
1  \\
\pm 1 \\
\vdots \\
\pm 1 
\end{matrix} \vline
& O (n^{-1} )
\end{pmatrix}; \qquad L_{2,1} = O(n^{-1}); \qquad L_{3,1} = O (n^{-1})
\]
The signs $\pm 1$ appearing in the first column of $L_{1,1}$ depend on the order of basis matrices and we do not require their exact values. Note also that one can show that $L_{2,1} = 0$, but we will not need this result in what follows. 

Next, we study $L_{1,2}$, $L_{2,2}$ and $L_{3,2}$.  
Since $\tilde T^{(n)}_\emptyset = I \not \in \mathcal X_2$, we know that $\tilde T_{\alpha}$ acts on $\mathcal X_2$ as
\[
\tilde T_{\alpha}^{(n)}  P^\otimes_{\pi_T} \otimes P^\otimes_{\pi_B} \begin{cases}
 = P^\otimes_{\pi_T} \otimes P^\otimes_{\pi_B} & \text{if $\alpha = \emptyset$} \\
\in {\rm span} (\mathcal X_3) & \text{otherwise} 
\end{cases} 
\]
Hence, $Q P^\otimes_{\pi_T} \otimes P^\otimes_{\pi_B} =  P^\otimes_{\pi_T} \otimes P^\otimes_{\pi_B} + X_3$, with $X_3 \in  {\rm span} (\mathcal X_3)$. Thus,
\[
L_{1,2} =0;\qquad L_{2,2} = I; \qquad \text{$L_{3,2}$ = (unspecified)} 
\]

Finally, using \eqref{eq:L-1-1}, we have 
\begin{equation}
\tilde T_\alpha^{(n)} \tilde T_\beta^{(n)}P^\otimes_{\pi_T} \otimes  P^\otimes_{\pi_B}=  \left[ \begin{cases} 
\tilde T_{\alpha \vee \beta}^{(n)} & \text{if $\alpha \vee \beta$ exists in the poset }(\hat S_r, \leq)\\
 O(n^{-1}) & \text{otherwise} 
\end{cases} \right ] P^\otimes_{\pi_T} \otimes  P^\otimes_{\pi_B}
\end{equation}
Since $I\not\in \mathcal X_3$ we have $\beta \neq \emptyset$ and then
it follows from \eqref{calculation:spectral} that 
\[
L_{1,3} = O(n^{-1}); \qquad L_{2,3}=O(n^{-1});\qquad  L_{3,3}= O(n^{-1})
\]

Therefore, as $n \rightarrow \infty$, the matrix representation of $L_Q$ converges to a lower semi-triangular matrix whose diagonal elements are $0$ or $1$, proving the theorem.
\end{proof}

As a corollary of the theorem, we obtain the asymptotic positivity of the matrices $Q_\alpha^{(n)}$.
\begin{corollary}\label{cor:Q-beta-asympt-positive}
For all $r \geq 1$ and $\alpha \in \hat{\mathcal S}_r$, the spectrum of the matrix $Q_\alpha^{(n)}$ is at a distance $O(1/n)$ from the set $\{0,1\}$.
\end{corollary}
\begin{proof}
This follows from equation \eqref{eq:decomp-Q-alpha}, Theorem \ref{thm:Q-asympt-positive} and the fact that first tensor factor in equation \eqref{eq:decomp-Q-alpha} is a rank-one projector. 
\end{proof}

\section{Output states for product of conjugate random quantum channels}\label{sec:main}

This section contains the main probabilistic result of the paper, a limit theorem for a model of random matrices. We investigate output quantum states of a product of
$r$ identical random quantum channels and their complex conjugates:
\be
\Psi_n^{\otimes r} = \Phi_n^{\otimes r} \otimes \bar\Phi_n^{\otimes r} 
\ee
Here, $\Psi_n = \Phi_n \otimes \bar\Phi_n$ where $\Phi_n$ are random quantum channels, as defined in Section \ref{sec:graphical-calculus}. Given a sequence of input states $\psi_n \in \C^{2rn}$, we are interested in the sequence of output states 
\be
Z_n = \left[\left(\Phi_n^C\right)^{\otimes r} \otimes \left(\bar\Phi_n^C\right)^{\otimes r}\right] (\psi_n \psi_n^*),
\ee
where we consider complementary channels $\Phi_n^C$. Note that, as it was noted in Section \ref{sec:graphical-calculus}, replacing a channel by its complementary does not change the (non-zero) spectrum of the outputs for pure inputs. This is important for us given the asymptotic regime we are interested in ($k$ fixed and $n \to \infty$), since the output matrix $Z_n$ lives in the fixed space $M_{k^{2r}}(\C)$.

In the graphical calculus, the sequence of vectors $\psi_n$ on $\C^{2rn}$ 
is represented as the $rn \times rn$ matrix $A_n$ 
through usual isomorphism between $\C^{2rn}$ and $M_{rn}(\C)$, see Figure \ref{fig:Zn} for the case $r=2$. 

\begin{figure}[htbp]
\includegraphics{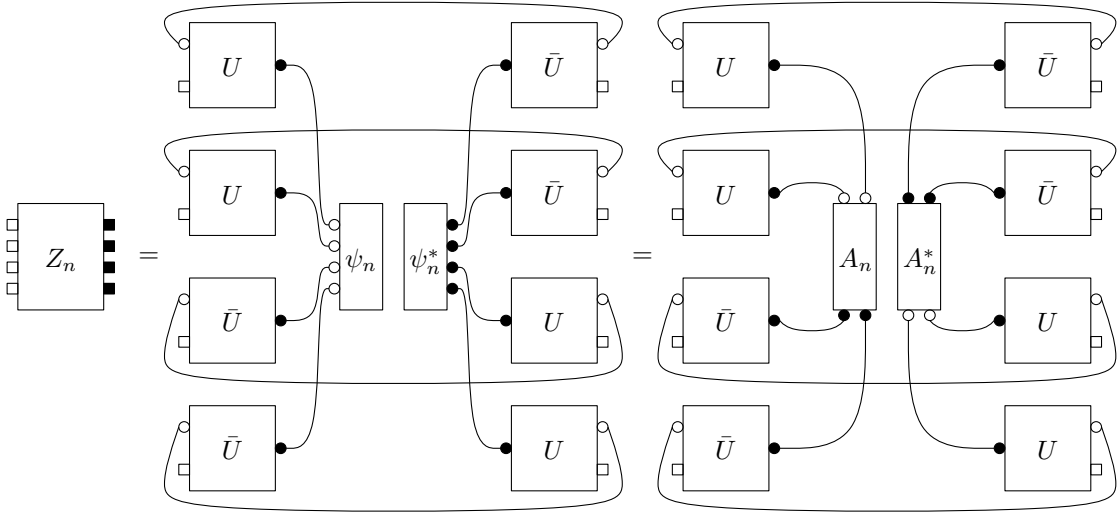}
\caption{Diagram for the output matrix $Z_n$ in the case $r=2$.} 
\label{fig:Zn}
\end{figure}

We are interested in computing the asymptotic moments $\Ex \trace Z_n^p$ of the output matrix $Z_n$. To do this, we are going to use the graphical Weingarten calculus, Theorem \ref{thm:graphical-Wg}. The diagram for $\trace Z_n^p$ can be thought of containing $p$ copies of the diagram in Figure \ref{fig:Zn} connected in a tracial manner.  This diagram contains $2rp$ boxes $U$, which shall be labeled by triples $[i,x,P]$, where
\begin{itemize}
\item $i \in \{1, \ldots, p\}$ indicates the index of the copy of $Z_n$ the $U$ box belongs to;
\item $x \in \{1, \ldots, r\}$ denotes the index of the channel $\Phi$ or $\bar \Phi$;
\item $P \in \{T,B\}$ indicates whether the box $U$ belongs to a $\Phi$ channel ($P=T$) or to a $\bar \Phi$ channel ($P=B$).
\end{itemize}
\begin{figure}[htbp]
\includegraphics{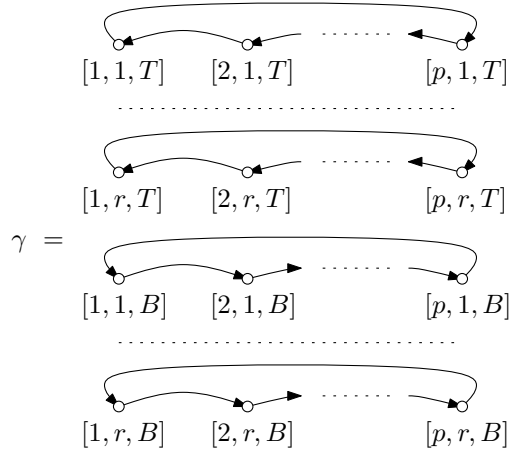}
\caption{A representation of the permutation $\gamma$ encoding the tracial structure of the moments. Note that in the top $r$ rows, $\gamma$ acts like a decreasing cycle, because the unitary operator $U$ appears on the left hand side of the channel $\Phi$ so it is connected to the $\bar U$ operator appearing in the previous block. In the bottom rows, the situation is reversed for the channel $\bar \Phi$.} 
\label{fig:gamma}
\end{figure}
Let $\delta \in \S_{2rp}$ be the permutation that swaps top and bottom indices
$$\delta = \prod_{i=1}^p\prod_{x=1}^r ([i,x,T] \, [i,x,B])$$
Also, define  $\gamma \in \S_{2rp}$ as the permutation that encodes the trace of the product of the $Z_n$ matrices, see Figure \ref{fig:gamma}:
$$
\gamma = \prod_{x=1}^r \left([r,x,T] \, [r-1,x,T] \cdots [1,x,T] \right) \prod_{y=1}^r \left([1,y,B] \, [2,y,B] \cdots [r,y,B] \right)
$$ 

Then, we have the following formula for the moments of  the random matrix $Z_n$:
\begin{theorem}\label{thm:main-moments} 
For a given sequence of inputs $\psi_n$ 
\be
\Ex \trace Z_n^p
=(1+ O(n^{-2}))   \sum_{\alpha_i \leq \beta_i } 
k^{-|\alpha^{-1}\gamma|} \prod_{i=1}^p 
\left[ 
\langle \psi_n | \tilde T_{\beta_i}^{(n)}  | \psi_n \rangle 
\left(-k^{-1}\right)^{|{\rm dom}(\beta_i)|-|{\rm dom}(\alpha_i)|}
\right]
\ee 
Here, the operators $\tilde T_{\beta_i}^{(n)}$ where defined in \eqref{operator_rescaled},  $\alpha_i,  \beta_i \in \hat \S_{r}$ are partial permutations, 
and $\alpha \in \S_{2rp}$ is a permutation induced from the $\alpha_i$'s in the following way:
\be\label{eq:def-alpha}
\alpha = \prod_{i=1}^p \prod_{x \in \mathrm{dom}(\alpha_i)}([i,x,T], [i, \alpha_i (x), B] )
\ee
\end{theorem} 
\begin{proof} 
First, the graphical Weingarten formula from Theorem \ref{thm:graphical-Wg} allows us to write $\Ex \trace[Z_n^p]$ as a sum over diagrams indexed by a pair of permutations $(\alpha, \beta) \in \S_{2rp}$. For such a pair, the corresponding diagram will contain (see Figure \ref{fig:diagram-U} for the wiring of a box $U$ appearing on the top):

\begin{figure}[htbp]
\subfigure[]{\includegraphics{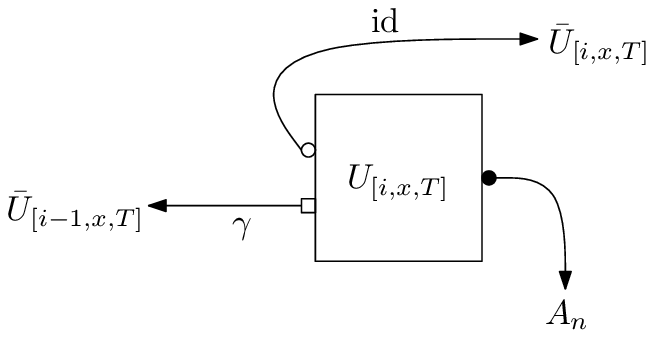}}\qquad
\subfigure[]{\includegraphics{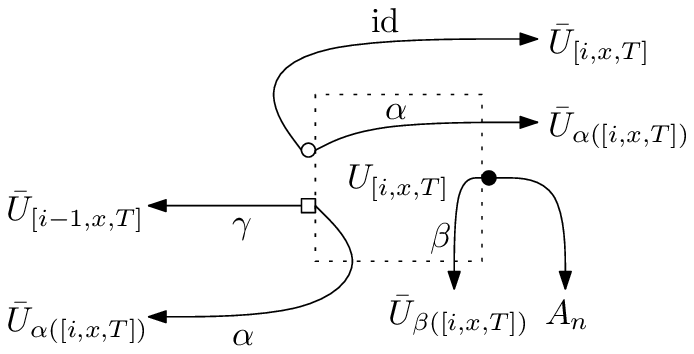}}
\caption{Wires attached to a $U$-box appearing on top ($\Phi$ channels) before and after the graph expansion are represented in (a), respectively (b).} 
\label{fig:diagram-U}
\end{figure}

\begin{enumerate}
\item $\# \alpha$ loops corresponding to $\C^n$, which give a contribution of $n^{\#\alpha}$;
\item $\#(\alpha^{-1}\gamma)$ loops corresponding to $\C^k$, which give a contribution of $k^{\#(\alpha^{-1}\gamma)}$;
\item \emph{necklace} diagrams containing wires and $A_n$ boxes, which give a contribution of $f_{A_n} (\beta)$.
\end{enumerate}
It follows that 
\be\label{eq:EtrZp}
\Ex \trace[Z_n^p] = \sum_{\alpha,\beta \in S_{2rp}} 
n^{\#\alpha} k^{\#(\alpha^{-1}\gamma)} f_{A_n} (\beta) \Wg(\alpha^{-1}\beta)
\ee
The contribution of $A_n$-necklaces $f_{A_n} (\beta)$ can be expressed using the formalism of generalized traces introduced in Appendix \ref{sec:trace-bounds}, see Figure \ref{fig:An} for a graphical proof of this fact:
\be
f_{A_n}(\beta) = \trace_{\beta^{-1}\delta}(A_n,A_n^*, \ldots, A_n,A_n^*)
\ee

\begin{figure}[htbp]
\includegraphics{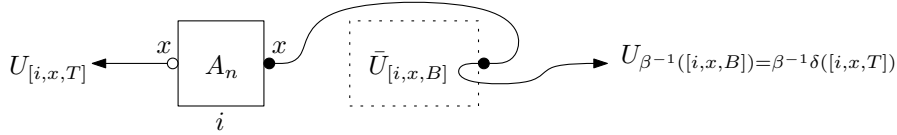}
\caption{The contribution of the necklaces containing the matrices $A_n$ and $A_n^*$ is given by the generalized trace of these matrices with respect to the permutation $\beta^{-1}\delta$.} 
\label{fig:An}
\end{figure}

Note that, in order to have the matrices $A_n$ and $A_n^*$ appearing in the right order, we have to order the triplets $[\cdot, \cdot, \cdot]$ as follows: $[i,x,P] < [j,y,Q]$ if and only if $i<j$ or ($i=j$, $P=T$ and $Q=B$) or ($i=j$, $P=Q$ and $x < y$). One can use  Lemma \ref{thm:bound-generalized-moment} to bound the generalized trace
\be
f_{A_n}(\beta) \leq n^{{\rm dist}(\beta^{-1}\delta,\Theta)} = n^{{\rm dist}(\beta^{-1},\Theta \delta)} 
\ee
To obtain this bound, we used the normalization condition $\|\psi_n\| = \|A_n\|_2 =1$. Recall that the set of permutations $\Theta$ is defined to be
\begin{align*}
\Theta &=  \!\!\!\!\! \bigcup_{\substack{E \sqcup F = \{1, \ldots, p\} \times \{T,B\} \\ |E| = |F| = p}}  \!\!\!\!\!  \left\{\theta \in \S_{2pr} \, : \, \forall e \in E \, \forall x \in \{1, \ldots , r\} , \, \theta([e,x]) = [f,*] \text{ for some } f \in F \right.\\
& \qquad \qquad \qquad \left. \text{ and } \forall f \in F \, \forall x \in \{1, \ldots , r\} , \, \theta([f,x]) = [e,*] \text{ for some } e \in E\right\}
\end{align*}

Hence, using $\#\alpha = 2rp - |\alpha|$ and the Weingarten asymptotic $\Wg(\alpha^{-1}\beta) \sim n^{ -2rp - |\alpha^{-1}\beta| }$, we have
\be \label{eq:power-inequality}
(\text{The power of $n$ in $\Ex \trace[Z_n^p]$}) 
&\leq& 2rp - |\alpha| + {\rm dist}(\beta^{-1}, \Theta\delta) -2rp - |\alpha^{-1}\beta| \\
&\leq&  \min _{\theta \in \Theta} |\beta \theta\delta| -|\beta| \leq \min _{\theta \in \Theta} | \theta \delta| =0 
\ee 
Here, we have used the triangle inequality twice $|\alpha| + |\alpha^{-1}\beta| \geq |\beta|$, $ |\beta \theta\delta| -|\beta| \leq  |\theta\delta|$ and the fact that $\delta \in \Theta$, which is obvious from the definition of $\Theta$, with the choice $E=E_T = \{1, \ldots, p\} \times \{T\}$ and $F=F_B = \{1, \ldots, p\} \times \{B\}$.

In order to find the contributing pairs $(\alpha, \beta)$, we have to investigate the equality cases in the above inequalities. 
For the first triangle inequality $|\alpha| + |\alpha^{-1}\beta| \geq |\beta|$, the bound is saturated if and only if $\alpha$ is on 
the geodesics between $\id$ and $\beta$, denoted by $\id \rightarrow\alpha\rightarrow\beta$.
For the second inequality $\min _{\theta \in \Theta} |\beta \theta\delta| -|\beta| \leq 0$, the equality case reads
\be
\min _{\theta \in \Theta} |\beta \theta \delta| =|\beta| 
\ee
To conclude, we aim at finding all pairs $(\alpha,\beta)$ such that
\begin{align}
\label{eq:triangle-equality-alpha}&\id \rightarrow\alpha\rightarrow\beta\\
\label{eq:triangle-equality-beta}&|\beta \theta \delta| \geq |\beta| \quad \text{for all $\theta \in \Theta$}
\end{align}

Consider a permutation $\sigma = \sigma_T \oplus \sigma_B \in \S_{2rp}$ where $\sigma_{T,B}$ leave invariant the top and bottom elements, i.e. 
\begin{align}
\sigma_T(\{1, \ldots, p\} \times \{1, \ldots, r\} \times \{T\}) &= \{1, \ldots, p\} \times \{1, \ldots, r\} \times \{T\}\\
\sigma_B(\{1, \ldots, p\} \times \{1, \ldots, r\} \times \{B\}) &= \{1, \ldots, p\} \times \{1, \ldots, r\} \times \{B\}
\end{align}
Then, one can check that $\sigma \delta \in \Theta$, for the choice $E=E_T$, $F=F_B$. In other words, for all $\sigma_T, \sigma_B$, we have $ \sigma_T \oplus \sigma_B \in \Theta\delta$. In particular,
\be
\forall i,j,x,y, \quad ([i,x,T],[j,y,T]), ([i,x,B],[j,y,B]) \in \Theta\delta
\ee
Note that we do not exclude the case $i =j$ or $x=y$.

Consider now a transposition $\tau  = ([i,x,T],[j,y,B])$ with $i \neq j$. We have
$$ \tau \delta = ( [i,x,B], [j,y,B], [j,y,T], [i,x,T] )  \prod_{\substack{l,z \\ (l,z) \notin \{(i,x),(j,y)\}}} ( [l,z,T] , [l,z,B] )$$
One can see that $\tau \delta \in \Theta$ for the choice
\begin{align}
E &= \{1, \ldots, p\} \times \{T\} \setminus \{(j,T)\} \sqcup \{(j,B)\} \\
F &= \{1, \ldots, p\} \times \{B\} \setminus \{(j,B)\} \sqcup \{(j,T)\}
\end{align}

We have thus shown that the following transpositions belong to $\Theta\delta$:
\begin{itemize}
\item $ ([i,x,T],[j,y,T]), ([i,x,B],[j,y,B])$, for all $i,j,x,y$;
\item $([i,x,T],[j,y,B])$ for all $i,j,x,y$ such that $i \neq j$.
\end{itemize}
On the other hand, it follows from \cite[Lemma 23.10]{nica-speicher} that $|\sigma \cdot (i,j)| = |\sigma| -1$ for any permutation $\sigma$ where $i,j$ belong to the same cycle. Hence, 
in order for \eqref{eq:triangle-equality-beta} to be satisfied, we see that the following pair of elements can not belong to the same cycle of $\beta$:
\begin{itemize}
\item $[i,x,T]$ and $[j,y,T]$, for all $i,j,x,y$;
\item $[i,x,B]$ and $[j,y,B]$, for all $i,j,x,y$;
\item $[i,x,T]$ and $[j,y,B]$ for all $i,j,x,y$ such that $i \neq j$.
\end{itemize}
Note that $\beta$ can not have cycles of length larger than 2, since in that case at least two of the elements in the cycle would have the same top or bottom index, contradicting one of the first two conditions above. It follows that $\beta$ is a product of disjoint transpositions (swapping top and bottom elements). The final condition above implies that these transpositions should swap elements belonging to the same ``$i$'' group. 
Therefore, the equation \eqref{eq:triangle-equality-beta} finally  implies that $\beta$ should be of the form
\be\label{beta-cand}
\beta = \prod_{i=1}^p \prod_{x \in {\rm dom} (\beta_i)} ([i,x,T], [i, \beta_i (x), B] )
\ee
where $\beta_i \in \hat \S_r$ for $i=1,\ldots, p$. 

Having solved the equation \eqref{eq:triangle-equality-beta}, we move on to finding $\alpha$ satisfying \eqref{eq:triangle-equality-alpha}. Since $\beta$ is a product of disjoint transpositions \eqref{beta-cand}, $\alpha$ should be constructed from a subset of the transpositions appearing in $\beta$. It follows that
\be\label{alpha-cand}
\alpha = \prod_{i=1}^p \prod_{x \in {\rm dom} (\alpha_i)} ([i,x,T], [i, \alpha_i (x), B] )
\ee
for partial permutations $\alpha_i \in \hat \S_r$ satisfying $\alpha_i \leq \beta_i$. 
Plugging the values for $\alpha$ and $\beta$ in equation \eqref{eq:EtrZp}, we obtain
\be
\Ex \trace Z_n^p
&=&(1+ O(n^{-2}))   \sum_{\alpha, \beta \text{ as in } \eqref{alpha-cand}, \eqref{beta-cand}} 
k^{\#(\alpha^{-1}\gamma)} 
\frac{f_{A_n}(\beta)}{n^{{\rm dist} (\beta^{-1},\Theta\delta)}} 
k^{-2rp -|\alpha^{-1}\beta|} (-1)^{|\alpha^{-1}\beta|}
\ee
The error term $1+O(n^{-2})$ comes from the Weingarten asymptotic \eqref{eq:Wg-asympt} and from the fact that the expression $2rp - |\alpha| + {\rm dist}(\beta^{-1}, \Theta\delta) -2rp - |\alpha^{-1}\beta|$ appearing in \eqref{eq:power-inequality} has a constant parity as a function of $(\alpha,\beta)$. Indeed, this follows from the fact that the $|\sigma \tau| = |\sigma| \pm 1$ for any permutation $\sigma$ and transposition $\tau$. Since both $\alpha$ and $\beta$ appear twice in the expression we are investigating, the parity conservation property follows.

We are going to further simplify the formula for $\Ex \trace Z_n^p$ by using the fact that $|\alpha^{-1}\beta| = \sum_i |\mathrm{dom}(\beta_i)| - |\mathrm{dom}(\alpha_i)|$ and, for $\beta$ as in \eqref{beta-cand}, 
$$f_{A_n}(\beta) = \prod_{i=1}^p \langle \psi_n | T_{\beta_i}^{(n)}  | \psi_n \rangle$$
Furthermore, note that for such $\beta$, ${\rm dist} (\beta^{-1},\Theta\delta) = \sum_{i=1}^p |{\rm dom} (\beta_i)|$, since with non trivial partial permutations $\beta_i$, $\beta\delta$ is not an element of $\Theta$, so in order to correct $\beta$ in such a way that $\beta \delta \in \Theta$, one has to undo the wires appearing in ${\rm dom} (\beta_i)$, for all $i$. With all the above ingredients, we obtain the announced formula for the $p$-th moment of the output matrix $Z_n$:
\be
\Ex \trace Z_n^p&=& (1+ O(n^{-2}))   \sum_{\alpha_i \leq \beta_i} 
k^{-|\alpha^{-1}\gamma|} \prod_{i=1}^p 
\left[ 
\frac{\langle \psi_n | T_{\beta_i}^{(n)}  | \psi_n \rangle}{n^{|{\rm dom}(\beta_i)|}}  
\left(-k^{-1}\right)^{|{\rm dom}(\beta_i)|-|{\rm dom}(\alpha_i)|}
\right]
\ee 
 
\end{proof}  

Since the matrix $Z_n$ is living in a space of fixed dimension, convergence of moments can be easily translated into the convergence of the random matrix itself. We define for $\beta \in \hat \S_r$
\be\label{operator:R}
R_\beta^{(k)} &=& \bigotimes _{x \in {\rm dom} (\beta)} [B_{x,\beta (x)}^{(k)}  - k^{-1} I] \otimes I \\
&=& \left[\sum_{A \subseteq \mathrm{dom}(\beta)} \left(\bigotimes _{x \in A} B_{x,\beta (x)}^{(k)}\right)  \otimes (- k^{-1})^{|{\rm dom}(\beta)|-|A|}I \right] \otimes I
 \notag\\
&=& \sum_{\alpha \leq \beta} T_\alpha^{(k)}  \cdot (-k^{-1})^{|{\rm dom}(\beta)|-|{\rm dom} (\alpha)|} \notag
\ee
where the second equality follows from the binomial formula. 

\begin{theorem}\label{thm:main-expectation}  
Consider a sequence of input states $\psi_n$. Then
$Z_n$ has the following expectation: 
\be
\Ex  Z_n 
=  (1+ O(n^{-2}))k^{-2r}  \sum_{\beta \in \hat S_r}  
\langle \psi_n | \tilde T_{\beta}^{(n)}  | \psi_n \rangle  R_\beta^{(k)},  
\ee
where the $O(n^{-2})$ error appears in each entry of the matrices
and $R_\beta^{(k)}$ is defined in \eqref{operator:R}.
\end{theorem} 
\begin{proof}
As in the previous theorem, we shall use the graphical Weingarten calculus, with the major difference that this time, we are averaging operators, and not scalars. Moreover, we are in the simplest possible case, $p=1$. By replacing $k^{\#(\alpha^{-1}\gamma)}$ by the $k^{2r}\times k^{2r}$ matrices $T^{(k)}_\alpha$ (see \eqref{operator:T}) in equation \eqref{eq:EtrZp}, we obtain
\be
\Ex Z_n = \sum_{\alpha,\beta \in S_{2r}} 
n^{\#\alpha} f_{A_n} (\beta) \Wg(\alpha^{-1}\beta) \cdot T^{(k)}_\alpha
\ee
Since $k$ is fixed, the terms which survive asymptotically are the same as before, so we conclude 
\begin{align*}
\Ex  Z_n
&= (1+ O(n^{-2}))  k^{-2r}  \sum_{\alpha \leq \beta} 
\left[ 
\langle \psi_n | \tilde T_{\beta}^{(n)}  | \psi_n \rangle 
\left(-k^{-1}\right)^{|{\rm dom}(\beta)|-|{\rm dom}(\alpha)|}
\right] T_\alpha^{(k)} \\
&=  (1+ O(n^{-2}))   k^{-2r}  \sum_{\beta \in \hat S_r}  
\langle \psi_n | \tilde T_{\beta}^{(n)}  | \psi_n \rangle  R_\beta^{(k)}  
\end{align*}
\end{proof}

In order to be able to state almost-sure convergence results for the random matrix $Z_n$, we have to make the following assumption on the behavior of the input sequence $\psi_n$.

\medskip
\noindent \textbf{Assumption on input vectors:} a sequence of vectors $(\psi_n)_n$ is called \emph{well behaved} if for all partial permutations $\beta \in \hat{\mathcal S}_r$, one has
\begin{equation}\label{eq:assumption-moments}
\lim_{n \to \infty} \langle \psi_n | \tilde T_{\beta}^{(n)} | \psi_n \rangle = a_\beta \in[0,1].
\end{equation}
\medskip

The set of numbers $(a_\beta)_{\beta \in \hat \S_r}$ will be treated as parameters in our model from now on. Obviously, one has the normalization condition $a_\emptyset = 1$. The following theorem is the main result of this section.

\begin{theorem}\label{thm:main} 
Given a sequence of well-behaved inputs, 
the output matrices $Z_n $ converge almost surely to 
\be\label{formula:Z}
Z=k^{-2r}  \sum_{\beta \in \hat S_r}  
a_\beta R_\beta^{(k)} 
\ee 
\end{theorem} 
\begin{proof}
We shall use the Hilbert-Schmidt (or the 2-Schatten) norm to prove the convergence.
Note however that all norms are equivalent on the finite-dimensional space $M_{k^{2r}}(\C)$.

It suffices to show that
\be
\sum_{n=1}^\infty \Ex \|Z_n - \Ex Z_n\|_2^2 < \infty
\ee
and then apply the Borel-Cantelli lemma to conclude that
$ \|Z_n - \Ex Z_n\|_2$ converges to zero almost surely. We have 
$$ \Ex \|Z_n - \Ex Z_n\|_2^2 = \Ex \trace(Z_n^2) -  \trace[(\Ex Z_n)^2]$$

The first term above has been computed in Theorem \ref{thm:main-moments}, for $p=2$ (recall that the permutation $\alpha$ is defined in \eqref{eq:def-alpha})
\be
\Ex \trace Z_n^2
=(1+ O(n^{-2}))   \sum_{\substack{\alpha_1 \leq \beta_1 \\ \alpha_2 \leq \beta_2}} 
k^{-|\alpha^{-1}\gamma|} \prod_{i=1}^2 
\left[ 
\langle \psi_n | \tilde T_{\beta_i}^{(n)}  | \psi_n \rangle 
\left(-k^{-1}\right)^{|{\rm dom}(\beta_i)|-|{\rm dom}(\alpha_i)|}
\right]
\ee 
The second term can be easily computed from Theorem \ref{thm:main-expectation}
\begin{align*}
\trace (\Ex Z_n)^2 
&= \trace  \left[(1+ O(n^{-2}))  
\left(k^{-2r}  \sum_{\alpha \leq \beta \in \hat \S_r} 
\left[ 
\langle \psi_n | \tilde T_{\beta}^{(n)}  | \psi_n \rangle 
\left(-k^{-1}\right)^{|{\rm dom}(\beta)|-|{\rm dom}(\alpha)|}
\right] T_\alpha^{(k)} \right)^2 \right]\\
&= (1+ O(n^{-2}))  
\sum_{\substack{\alpha_1 \leq \beta_1 \\ \alpha_2 \leq \beta_2}} k^{-4r}\trace \left[ T_{\alpha_1}^{(k)} T_{\alpha_2}^{(k)} \right]
 \prod_{i=1}^2
\left[ 
\langle \psi_n | \tilde T_{\beta_i}^{(n)}  | \psi_n \rangle 
\left(-k^{-1}\right)^{|{\rm dom}(\beta_i)|-|{\rm dom}(\alpha_i)|}
\right]
\end{align*}
Note that one can compute
\be
k^{-4r}\trace \left[ T_{\alpha_1}T_{\alpha_2}\right] = k^{-4r + \# (\alpha^{-1} \gamma)} = k^{-|\alpha^{-1} \gamma|}
\ee
where $\alpha  \in \S_{4r}$ is defined as before. We thus get the estimate which allows us to conlcude:
\be
\Ex \|Z_n - \Ex Z_n\|_2^2  =\Ex \trace(Z_n^2) -  \trace[(\Ex Z_n)^2] = O(n^{-2})
\ee 
\end{proof}

\section{Optimality of products of Bell states}\label{sec:Bell-optimal}

In this section, we show that among the well-behaved input states satisfying assumption \eqref{eq:assumption-moments}, the ones having minimal output entropy for generic random channels are tensor products of Bell (or maximally entangled) states. In particular, we recover results from \cite{cfn1} in the case $r=1$.

First, the conclusion of our main result from Section \ref{sec:main}, Theorem \ref{thm:main}, can be reformulated as follows. 

\begin{proposition}\label{prop:Z-conv-comb}
For a fixed sequence of well-behaved input states, the output matrix $Z_n$ converges almost surely to the matrix
\begin{equation}\label{eq:decomposition-Z-alpha}
Z=\sum_{\alpha \in \hat{\mathcal S}_r} p_\alpha Z_\alpha^{(k)},
\end{equation}
where $p_\alpha$ are positive numbers that sum up to one and
\begin{align}\label{eq:def-Z-alpha}
Z_\alpha^{(k)} 
&= \left[ \bigotimes _{x \in {\rm dom}(\alpha)} C_{x, \alpha(x)}\right] \otimes \hat I 
\end{align}
Here, $\hat I$ is the identity operator normalized to have unit trace and 
$$C_{x,y} = k^{-2}B_{x,y}^{(k)} + (k^{-2} - k^{-3}) I$$
\end{proposition}
\begin{proof}
One can rewrite $Z_\beta^{(k)}$ as
\begin{align*}
Z_\beta^{(k)} 
&=k^{-2r} \bigotimes _{x \in {\rm dom} (\beta)} 
\left[ \left(B_{x, \beta(x)}^{(k)} -k^{-1} I \right) + I \right] \otimes I \\  
&=k^{-2r} \sum_{\alpha \leq \beta} R_\alpha^{(k)}
\end{align*}
Here, the last equality comes from \eqref{operator:R}. 
Then, using the M\"obius inversion formula for the poset $\hat{\mathcal S}_r$, 
we can express the $R$ matrices in terms of the $Z$ matrices:
$$R_\beta^{(k)} = k^{2r} \sum_{\alpha \leq \beta} (-1)^{|{\rm dom}(\beta)|-|{\rm dom}(\alpha)|}Z_\alpha^{(k)}$$

Plugging this expression into the limiting matrix formula in \eqref{formula:Z},  
we obtain
\begin{align*}
Z=\sum_{\alpha \leq \beta} (-1)^{|{\rm dom}(\beta)|-|{\rm dom}(\alpha)|} a_\beta Z_\alpha^{(k)}
= \sum_{\alpha \in \hat{\mathcal S}_r} p_\alpha Z_\alpha^{(k)},
\end{align*}
where 
$$p_\alpha = \sum_{\beta \geq \alpha} (-1)^{|{\rm dom}(\beta)|-|{\rm dom}(\alpha)|} a_\beta$$
All there is left to show is the positivity of the real numbers $p_\alpha$. Using the definition \eqref{eq:assumption-moments} of the coefficients $a_\beta$, we write
$$p_\alpha = \lim_{n \to \infty} \langle \psi_n | Q_\alpha^{(n)} | \psi_n \rangle,$$
where 
$$ Q_\alpha^{(n)} = \sum_{\beta \geq \alpha} (-1)^{|{\rm dom}\beta|-|{\rm dom}\alpha|} \tilde T_{\beta}^{(n)}$$
Since it was shown in Corollary \ref{cor:Q-beta-asympt-positive} that the operators $Q_\beta^{(n)}$ are asymptotically positive, the positivity of the $p_\beta$ follows and the proof is complete.

\end{proof}

The entropy of the density matrices $Z_\alpha^{(k)}$ are easily computed ($C$ is any of the matrices $C_{x,y}$):
\begin{equation}\label{eq:entropy-Z-alpha}
H(Z_\alpha^{(k)}) = |\mathrm{dom}(\alpha)| H(C) + (r-|\mathrm{dom}(\alpha)|) \log(k^2),
\end{equation}
Note that $H(C)$ is a constant depending on $k$, 
\begin{align}\label{eq:entropy-hc}
H(C) &= -(k^{-1} + k^{-2} - k^{-3}) \log (k^{-1} + k^{-2} - k^{-3}) - (k^2-1) (k^{-2} - k^{-3}) \log (k^{-2} - k^{-3}) \\
&= h(k^{-1} + k^{-2} - k^{-3}) + (k^2-1) h( k^{-2} - k^{-3})\\
& < \log(k^2),
\end{align}
where $h(x)  = -x \log x$. It follows that the entropy of  $Z_\alpha^{(k)}$ is a strictly decreasing function of $ |\mathrm{dom} (\alpha)|$. We now state the main result of the current section.

\begin{theorem}\label{thm:product-inputs}
Among all sequences of well-behaved input states, the ones having a minimal output entropy are the ones having parameters
$$a_\beta= 1_{\beta \leq \pi},$$
for some (full) permutation $\pi \in \mathcal S_r$. In this case, the input state is (asymptotically) a tensor product of Bell states (where the matching $\Phi \leftrightarrow \bar \Phi$ of the conjugate channels is given by $\pi$), the output state is $Z = Z_\pi$ and the output entropy is 
$$H(Z) = r H(C)$$
\end{theorem}
\begin{proof}
Using the concavity of the von Neumann entropy \cite[11.3.5]{nielsen-chuang} and equation \eqref{eq:decomposition-Z-alpha}, one has
$$H(Z) \geq \sum_{\alpha \in \hat{\mathcal S}_r} p_\alpha H(Z_\alpha^{(k)}).$$
The conclusion follows from the fact that the terms with full $\alpha$ have the least entropy. The unicity of the minimizer comes from the strict concavity of the entropy and form the fact that the density matrices $Z_\alpha^{(k)}$ are different. The rest of the statements in the theorem are trivial.
\end{proof}

\section{High entropy outputs and GHZ inputs}\label{sec:GHZ}

In this section we discuss a class of input states which give maximally mixed outputs, from which no information can be extracted. Although such examples are not interesting for the purpose of communicating classical information, they have theoretical interest.  After stating the main result in the following proposition, we discuss the particular cases of GHZ states and generic multi-partite pure states.

\begin{proposition}\label{prop:bad-inputs}
Consider a family $\psi_n$ of normalized input states such that, for all $\beta \in \hat{\mathcal S}_r$, $\beta \neq \emptyset$,
$$\langle \psi_n | \tilde T_\beta^{(n)} | \psi_n\rangle = o(1).$$
Then, the output state is asymptotically maximally mixed, i.e.
$$\text{a.s.}, \qquad \lim_{n \to \infty} Z_n =  \frac{\mathrm{I}_{k^{2r}}}{k^{2r}}.$$
\end{proposition}
\begin{proof}
From Theorem \ref{thm:main}, the output state is asymptotically a mixture of the operators $R_\beta$, with coefficients given by the overlaps between the input vector $\psi_n$ and the operators $\tilde T_\beta^{(n)}$. Our assumption implies that all these coefficients will vanish asymptotically, except for the one corresponding to $\beta = \emptyset$. Hence, the almost sure limit of the output state $Z_n$ is
$$k^{-2r} R_\emptyset = \frac{\mathrm{I}_{k^{2r}}}{k^{2r}},$$
as announced.
\end{proof}

We now analyze the particular case of a GHZ input state \cite{ghz}. Such a state is defined by
\begin{equation}
(\mathbb C^n)^{2r} \ni \psi_\text{GHZ} = \frac{1}{\sqrt n} \sum_{i=1}^n \underbrace {e_i \otimes e_i \otimes \cdots \otimes e_i}_{2r \text{ factors}} .
\end{equation}

\begin{figure}[htbp]
\includegraphics{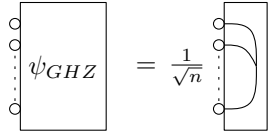}
\caption{Graphical representation of the GHZ input state.} 
\label{fig:GHZ-input}
\end{figure}

For a graphical picture of such a state, see Figure \ref{fig:GHZ-input}. Either by direct algebraic calculation of from graphical considerations, it is easy to see that one has, for all $\beta \in \hat{\mathcal S}_r$,
$$\langle \psi_\text{GHZ} | \tilde T_\beta^{(n)} | \psi_\text{GHZ} \rangle = n^{-\mathrm{dom}(\beta)}$$
The GHZ state satisfies thus the hypothesis of Proposition \ref{prop:bad-inputs} and we obtain the following corollary. 

\begin{corollary}
The output of the GHZ state $\psi_\text{GHZ}$ through a product of $r$ quantum channels and their conjugates
$$Z_n = [\Phi^{\otimes r} \otimes \bar{\Phi}^{\otimes r}](\psi_\text{GHZ}\psi_\text{GHZ}^*)$$
converges, almost surely, to a maximally mixed state.
\end{corollary}

We study now random pure states $\psi_n$ distributed uniformly on the unit sphere of $(\C^n)^{\otimes 2r}$. 

 \begin{proposition} 
 For any partial permutation $\beta \neq \emptyset$, with overwhelming probability, a random pure input state $\psi_n$ satisfies
 $$\langle \psi_n | \tilde T_\beta^{(n)} | \psi_n\rangle = o(1)$$
 \end{proposition}
 \begin{proof}
 First, note that 
 $$\Ex \langle \psi_n | \tilde T_\beta^{(n)} | \psi_n\rangle = \trace (\tilde T_\beta^{(n)} \Ex \psi_n \psi_n^*) = n^{-2r} \trace(\tilde T_\beta^{(n)}) = n^{-2\mathrm{dom}(\beta)}$$
 We shall use (see \cite{msc})
 \begin{lemma}[Levy's lemma]
 Let $f : S^{d-1} \to \R$ be a function defined on the unit sphere of $\R^d$ with Lipschitz constant $L$. Then 
 $$\P( |f-\Ex f| > \epsilon ) \leq \exp(-Cd\epsilon^2/L^2)$$
 where the expectation is taken with respect to the uniform measure on $S^{d-1}$ and $C$ is a constant.
 \end{lemma}
 \noindent with $d=2n^{2r}$ and $f(\psi) = \langle \psi_n | \tilde T_\beta^{(n)} | \psi_n\rangle$. Since $\tilde T_\beta^{(n)} $ is a projector, the function $f$ has Lipschitz constant bounded by 2. Putting $\epsilon = n^{-r+\delta}$, for some $\delta >0$, we obtain the announced result.
 \end{proof}

When contrasting the above results with the one in Theorem \ref{thm:product-inputs}, one concludes that the entanglement present in GHZ or generic states is not suitable for producing low-entropy outputs. The structure of the entanglement in the states from Theorem \ref{thm:product-inputs} seems to be essential in obtaining such low-entropy states.

\section*{acknowledgment}
Both authors thank Piotr \'Sniady for useful discussions and for his suggestion on how to deal with highly degenerate matrices. 
I.~N.~ would like to thank P\'eter Vrana for interesting discussions and the Technische Universit\"{a}t M\"{u}nchen for its hospitality during the summer months of 2012 when this project was initiated. The research of M.~F.~ was financially supported by the CHIST-ERA/BMBF project CQC. I.~N.~ acknowledges financial support from the ANR project OSvsQPI 2011 BS01 008 01.

\appendix\section{Bounds for generalized traces of matrices acting on tensor products}\label{sec:trace-bounds}

In this appendix we derive bounds for generalized traces of tensors in terms of their Schatten norms. These results are in the spirit of those obtained by Mingo and Speicher in \cite{mingo-speicher}, with two notable differences. We consider general tensors, whereas in \cite{mingo-speicher}, the authors investigate generalized traces of matrices. On the other hand, the type of traces we look at are less general than the ones in \cite{mingo-speicher}. 

In the current paper, we shall only use the $L^2$ incarnation of the results presented in this appendix. However, we think that the results are interesting on their own and might prove to be useful in other circumstances.  

Consider a set of $k$ matrices $A_1, \ldots, A_k$ acting on $(\mathbb C^n)^{\otimes r}$. Graphically, these matrices are represented by boxes with $r$ legs on each side. Let $\sigma \in \mathcal S_{kr}$ be a permutation that will be used to contract the $2kr$ legs of the boxes $A_i$. We shall be implicitly using the bijection $\{1, \ldots, kr\} = \{1, \ldots, k\} \times \{1, \ldots, r\}$ in such a way that elements in $\{1, \ldots, kr\}$ shall be denoted by pairs $(i,x)$, with $i \in \{1, \ldots, k\}$ and $x \in \{1, \ldots, r\}$. We call a \emph{generalized trace} of these matrices the quantity (see Figure \ref{fig:generalized-moment} for a graphical representation):
\begin{equation}\label{eq:def-generalized-moment}
\mathrm{Tr}_\sigma(A_1, \ldots, A_k)=\mathrm{Tr}\left[\left(\bigotimes_{i=1}^k A_i \right)P^\otimes_\sigma \right]
\end{equation}

\begin{figure}[htbp]
\includegraphics{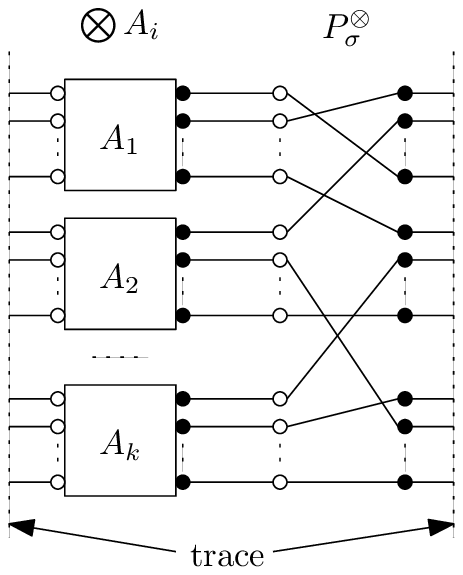}
\caption{Diagram for generalized trace of the matrices $A_1, \ldots, A_k$ connected with the permutation $\sigma$. The dotted vertical lines on each side of the diagram are identified and correspond to the trace in equation \eqref{eq:def-generalized-moment}.} 
\label{fig:generalized-moment}
\end{figure}

As a working example, let us consider the following example of a generalized trace:
\begin{equation}\label{eq:bound-example}
\mathrm{Tr} \left\{ [\mathrm{Tr} \otimes \mathrm{id}](A_1) \cdot [\mathrm{id} \otimes \mathrm{Tr} ](A_2)\right\}
\end{equation}
In Figure \ref{fig:bound-example-diagram}, we represent the trace using the graphical formalism from Section \ref{sec:graphical-calculus}. The same calculus can be represented as a trace of the tensor product of the matrices against a tensor permutation matrix, see Figure \ref{fig:bound-example-moment}. The permutation appearing in this example is simply the transposition $((1,2),(2,1))$, swapping the second leg of $A_1$ with the first leg of $A_2$.
\begin{figure}[htbp]
\subfigure[]{\includegraphics{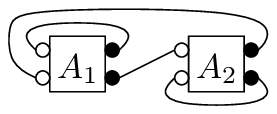}\label{fig:bound-example-diagram}}\qquad\qquad
\subfigure[]{\includegraphics{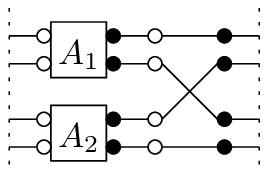}\label{fig:bound-example-moment}}\qquad\qquad
\subfigure[]{\includegraphics{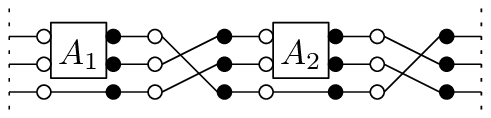}\label{fig:bound-example-input-output}}
\caption{The same diagram, represented in the usual graphical formalism (left), as a generalized trace formalism (center) and in the input-output formalism (right). The two dotted lines in each picture are identified and correspond to traces. The input-output presentation on the right has overhead $s=1$.} 
\label{fig:bound-example}
\end{figure}

Before we state and prove the main result, let us introduce two subsets $\Gamma$ and $\Theta$ of the permutation group $\mathcal S_{kr}$ which correspond to two special classes of generalized traces. 

First, the subset $\Gamma$ is defined such that the generalized trace \eqref{eq:def-generalized-moment} can be written 
as trace of the  product of $k$ factors of the type $A_i P^\otimes_{\rho_i}$ for $\rho_i \in \S_r$, where the order does not matter. 
Graphically, $\Gamma$ corresponds to the generalized trace which can be rearranged into 
a ``stream'' of matrices $A_i$'s with wires connecting the neighboring matrices.
So, this may be called ``input-output representation''.
Formally, for a permutation $\alpha \in \mathcal{S}_k$, let
$$\mathcal{S}_{kr} \supseteq [ \alpha ] = \{\beta \in \mathcal{S}_{kr} \, | \, \forall i \in \{1, \ldots, k\} , \forall x \in \{1, \ldots, r\}, \exists y \in \{1, \ldots, r\} \text{ s.t. } \beta(i,x)=(\alpha(i), y)\}$$
be the set of permutations of $\{1, \ldots, kr\}$ which preserve the blocks of size $r$ and act like $\alpha$ globally on these blocks. Then, 
$$\Gamma = \bigcup_{\substack{ \gamma \in \mathcal S_k \\ \#\gamma=1}} [\gamma]$$
Equivalently, recall that the wreath product $\mathcal S_r \wr \mathcal S_k$ has elements $(\beta;\alpha_1, \ldots, \alpha_k)$ with $\beta \in \mathcal S_k$ and $\alpha_i \in \mathcal S_r$. One has then
$$\Gamma = \{(\beta;\alpha_1, \ldots, \alpha_k) \in \mathcal S_r \wr \mathcal S_k, \, | \,  \#\beta=1\}$$

Secondly, when $k$ is even 
we define the subset $\Theta$ to be such that the generalized trace \eqref{eq:def-generalized-moment} 
can be written as trace of  product of ``rotated two columns'' where each column 
is a tensor product of $k/2$ matrices $A_i$. 
More precisely, 
the two-column structure of $\Theta$ is encoded into an equi-partition $\{1, \ldots, k\} = E \sqcup F$, 
with $|E| = |F| = k/2$ and we ask for a permutation $\theta \in \Theta$ to satisfy
\begin{align*}
\forall i \in E,\, \forall x \in \{1, \ldots, r\}, \quad \theta((i,x)) &= (j, y), \quad \text{for some } j \in F \text{ and } y \in \{1, \ldots, r\}\\
\forall j \in F,\, \forall x \in \{1, \ldots, r\}, \quad \theta((j,x)) &= (i, y), \quad \text{for some } i \in E \text{ and } y \in \{1, \ldots, r\}
\end{align*}

\begin{theorem}\label{thm:bound-generalized-moment}
For any permutation $\sigma \in \mathcal S_{kr}$ and any matrices $A_1, \ldots, A_k \in M_{n^r}(\mathbb C)$, the following bounds hold. 
\begin{enumerate}[(a)]
\item $L^1$-bound:
\begin{align}
\label{eq:bound-generalized-moment-1}\left\vert \mathrm{Tr}_\sigma(A_1, \ldots, A_k) \right \vert &\leq  \prod_{i=1}^{k} \|A_i\|_1 
\end{align} 
\item $L^\infty$-bound:
\begin{align} 
\label{eq:bound-generalized-moment-infty} \left\vert \mathrm{Tr}_\sigma(A_1, \ldots, A_k) \right \vert &\leq n^{r+\mathrm{dist}(\sigma, \Gamma)} \prod_{i=1}^{k} \|A_i\|_\infty
\end{align}
\item $L^2$-bound: if $k$ is even, then 
\begin{align}
\label{eq:bound-generalized-moment-2} \left\vert \mathrm{Tr}_\sigma(A_1, \ldots, A_k) \right \vert &\leq n^{\mathrm{dist}(\sigma, \Theta)} \prod_{i=1}^{k} \|A_i\|_2
\end{align}
\end{enumerate} 
\end{theorem}

\begin{proof}
The $L^1$ case follows trivially from the fact that the Schatten 1-norm (and all the other $p$-norms, for that matter) are unitarily invariant and multiplicative under tensor products. See Figure \ref{fig:generalized-moment}. Note that equality can be achieved in this case by taking 
$A_i = \bigotimes_{j=1}^r xx^*$ for the same unit vector $x \in \C^n$.

Let us now treat the $L^\infty$ bound, which is conceptually more interesting, and discuss the $L^2$ case last. To this end, we introduce the important definition of an \emph{input-output presentation} of the generalized trace given by $(A_1, \ldots, A_k;\sigma)$: it is an operator $T:(\mathbb C^n)^{\otimes (r+s)} \to (\mathbb C^n)^{\otimes (r+s)}$ of the form
\begin{equation}\label{eq:presentation_T}
T = \prod_{i=1}^k \left( A_{\tau(i)} \otimes \mathrm{I}_{n^s} \right) P^\otimes_{\rho_i}, 
\end{equation}
where $\tau \in \mathcal{S}_k$ is a permutation encoding the order of the matrices $A_i$ in the expansion, $P^\otimes_{\rho_i}$ are tensor permutation matrices acting on $(\mathbb C^n)^{\otimes (r+s)}$, with $\rho_i \in \mathcal{S}_{r+s}$. Figure \ref{fig:bound-example-input-output} contains an input-output presentation of the generalized trace appearing in \eqref{eq:bound-example}. The integer $s \geq 0$ is called the \emph{overhead} of the presentation. We ask for such a presentation to encode the generalized trace 
$$\mathrm{Tr}(T) = \mathrm{Tr}\left[\left(\bigotimes_{i=1}^k A_i \right)P^\otimes_\sigma \right]$$
Such representations have the advantage that they readily give bounds
\begin{align}
|\mathrm{Tr}(T) | &\leq  \|T\|_1 \leq n^{r+s} \|T\|_\infty \notag \\
&\leq n^{r+s} \prod_{i=1}^k \|A_{\tau(i)} \otimes \mathrm{I}_{n^s}\|_\infty \|P^\otimes_{\rho_i}\|_\infty \notag \\
&\leq n^{r+s}  \prod_{i=1}^{k} \|A_i\|_\infty \label{eq:moment-bound-presentation-2} 
\end{align}

We shall prove, by induction on the number $\mathrm{dist}(\sigma, \Gamma)$, that one can find an input-output presentation of a generalized trace with an overhead $s=\mathrm{dist}(\sigma, \Gamma)$. This is the key idea of the proof.

Let us first consider the case when $\sigma \in \Gamma$. This means that, up to permutations $\rho_i \in \mathcal{S}_r$  of the legs of each individual box $A_i$, the generalized trace is a trace of the product of all the matrices $A_i$, in some specific order given by a permutation $\tau$. In other words, one has
$$ \mathrm{Tr}_\sigma(A_1, \ldots, A_k) = \prod_{i=1}^k  A_{\tau(i)} P^\otimes_{\rho_i}, $$
and this is an input-output presentation of $\sigma$, without any overhead. This proves the initialization step of the induction, $\mathrm{dist}(\sigma, \Gamma)=0$. 

The inductive step of our claim is a consequence of the following lemma. 

\begin{lemma}\label{lem:surgery}
If a permutation $\sigma \in \mathcal{S}_{kr}$ admits an input-output presentation with overhead $s$ and $\xi$ is a transposition, then $\sigma \xi$ admits an input-output presentation with overhead $s+1$. 
\end{lemma}
\begin{proof}
Let $\xi = ((j_1,y_1),(j_2,y_2))$ be the transposition in the statement. If $j_1=j_2$, there is nothing to show, since the transposition $\xi$ can be absorbed in the tensor permutation matrix $P^\otimes_{\tau^{-1}(j_1)}$ in the presentation of $\sigma$. This way, one obtains a presentation for $\sigma \xi$, with overhead $s$ (hence it is possible with overhead $s+1$).

Suppose now $j_1\neq j_2$ and let $T$ be a presentation for $\sigma$, as it is defined in \eqref{eq:presentation_T}. Let also $k_1 = \tau^{-1}(j_1)$ and $k_2 = \tau^{-1}(j_2)$. The presentation $T$ can not be used directly to produce an input-output presentation for $\sigma \xi$, since one of the legs of the matrix $A_{k_1}$ needs to be connected to a different permutation and this is not allowed in the definition of an input-output presentation, see Figure \ref{fig:presentation-sigma-xi-bad}.

We shall modify the presentation $T$ for $\sigma$, by adding a unit of overhead, into a presentation $T'$ for $\sigma \xi$, as follows:
$$T' = \prod_{i=1}^k \left( A_{\tau'(i)} \otimes \mathrm{I}_{n^{s+1}} \right) P^\otimes_{\rho'_i},$$
with 
\begin{itemize}
\item $\tau' = \tau$ : the order of the matrices $A_i$ does not change;
\item For all $i \notin \{k_1,k_2\}$, $\rho'_i(x) = \rho(x) \, \forall x \in [r+s]$ and $\rho'_i(r+s+1) = r+s+1$ : for all boxes not affected by $\xi$, the tensor permutation matrices are the same;
\item $\rho'_{k_1}(y_1) = r+s+1$, $\rho'_{k_1}(r+s+1) = \rho_{k_1}(y_1)$ and  $\rho'_{k_1}(x) = \rho_{k_1}(x)$ for all $x \notin \{y_1,r+s+1\}$;
\item $\rho'_{k_2}(y_2) = r+s+1$, $\rho'_{k_2}(r+s+1) = \rho_{k_2}(y_2)$ and  $\rho'_{k_2}(x) = \rho_{k_2}(x)$ for all $x \notin \{y_2,r+s+1\}$.
\end{itemize} 

In other words, the extra copy of the space $\mathbb C^n$ is used to implement the transposition $\xi$ in a way compatible with the input-output structure of the presentation $T$, see Figure \ref{fig:presentation-sigma-xi} for a graphical representation of this procedure.

\begin{figure}[htbp]
\subfigure[]{\includegraphics{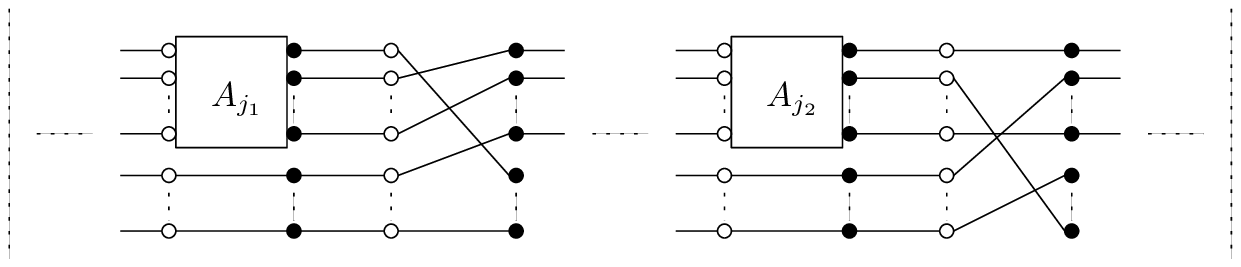}}\qquad\qquad
\subfigure[]{\label{fig:presentation-sigma-xi-bad} \includegraphics{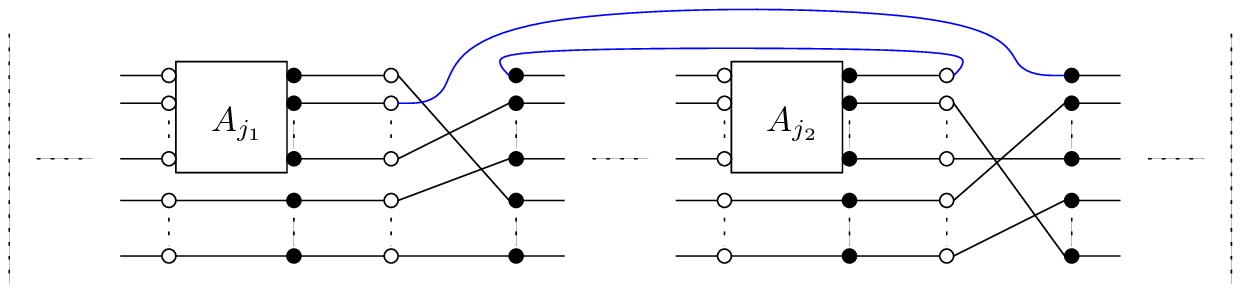}}\qquad\qquad
\subfigure[]{\includegraphics{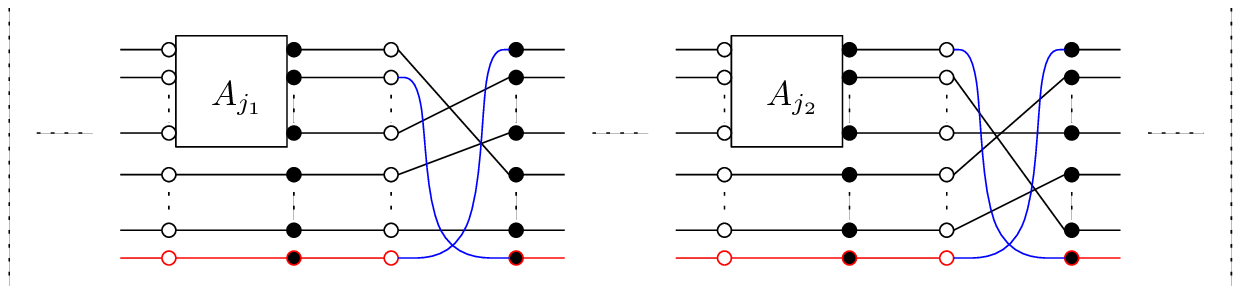}}
\caption{Input-output presentations for generalized traces. A presentation for $\sigma$ is depicted in the top figure. The center figure corresponds to $\sigma\xi$, but it is not an input-output presentation for this permutation, since it contains wires which do not connect consecutive blocks. The bottom figure depicts the presentation $T'$, obtained by adding an extra copy of $\mathbb C^n$ and implementing the transposition $\xi$ through this space.} 
\label{fig:presentation-sigma-xi}
\end{figure}

\end{proof}

Using the above lemma, one can produce an input-output presentation of a permutation inductively from a decomposition
$$\sigma = \gamma \xi_1 \xi_2 \cdots \xi_s,$$
where $\gamma \in \Gamma$, $s= \mathrm{dist}(\sigma, \Gamma)$ and $\xi_i$ are transpositions. Start with a zero-overhead presentation of $\gamma$ and add the transpositions $\xi_i$, each adding a unit of overhead to the presentation. At the end, one obtains a presentation for $\sigma$ with overhead $\mathrm{dist}(\sigma, \Gamma)$, and the $L^\infty$ bound  \eqref{eq:bound-generalized-moment-infty} follows.

(c) We now move to the $L^2$ bound \eqref{eq:bound-generalized-moment-2}. The proof strategy is exactly the same as in the $L^\infty$ case and we sketch only the differences. We introduce \emph{two-column presentations} as operators  $S:(\mathbb C^n)^{\otimes (kr/2+s)} \to (\mathbb C^n)^{\otimes (kr/2+s)}$ of the form
\begin{equation}\label{eq:presentation_S}
S = \prod_{j=1}^2 \left[ \left( \bigotimes_{i=1}^{k/2}A_{\tau_j(i)} \otimes \mathrm{I}_{n^s} \right) P^\otimes_{\rho_j}\right],
\end{equation}
where $\tau_1, \tau_2 : [k/2] \rightarrow [k]$  encode the order of the matrices $A_i$ in the presentation, and $P^\otimes_{\rho_j}$ are tensor permutation matrices acting on $(\mathbb C^n)^{\otimes (kr/2+s)}$, with $\rho_j \in \mathcal{S}_{kr/2+s}$. 
Figure \ref{fig:two-column-presentation} contains a two-column presentation of some generalized trace. 
Again, the integer $s \geq 0$ is called the \emph{overhead} of the presentation. As before, we ask that a presentation encodes the generalized trace, for any choice of $A_i$:
$$\mathrm{Tr}(S) = \mathrm{Tr}\left[\left(\bigotimes_{i=1}^k A_i \right)P^\otimes_\sigma \right].$$

\begin{figure}[htbp]
\includegraphics{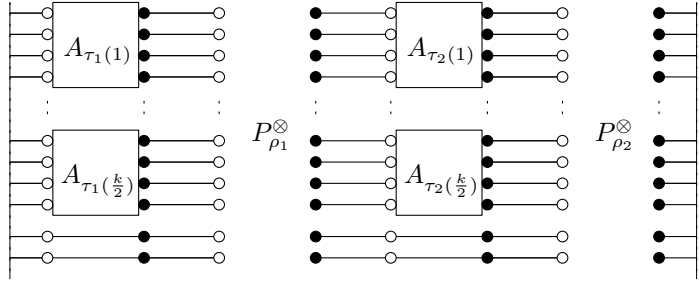}
\caption{A two-column presentation with overhead $s=2$ of a generalized trace, with $r=4$.} 
\label{fig:two-column-presentation}
\end{figure}

Using the Cauchy-Schwarz inequality, one can obtain the following bound:
\begin{align}
|\mathrm{Tr}(S) | &\leq   \left\|  \left( \bigotimes_{i=1}^{k/2}A_{\tau_1(i)} \otimes \mathrm{I}_{n^s} \right) P^\otimes_{\rho_1}\ \right\|_2 \left\|  \left( \bigotimes_{i=1}^{k/2}A_{\tau_2(i)} \otimes \mathrm{I}_{n^s} \right) P^\otimes_{\rho_2}\ \right\|_2 \notag \\
&\leq n^{s/2} \prod_{i=1}^{k/2} \|A_{\tau_1(i)} \|_2  \, n^{s/2} \prod_{i=1}^{k/2} \|A_{\tau_2(i)} \|_2\notag \\
&\leq n^{s}  \prod_{i=1}^{k} \|A_i\|_2. 
\end{align}

Since permutations $\sigma \in \Theta$ admit two-column presentations with overhead $s=0$, the bound \eqref{eq:bound-generalized-moment-2} holds in this case. For general $\sigma$, one proceeds as in the $L^\infty$ case, using Lemma \ref{lem:surgery} for each transposition $\xi_i$ appearing in the decomposition 
$$\sigma = \theta \xi_1 \xi_2 \cdots \xi_s,$$
where $\theta \in \Theta$.
\end{proof}

\begin{remark}
The case $r=1$ of the $L^\infty$ bound \eqref{eq:bound-generalized-moment-infty} is a trivial corollary of the result in \cite{mingo-speicher}, since the forest associated to the graph contains $\#\sigma$ trivial trees. 
\end{remark}

\end{document}